\documentclass[11pt]{article}
\usepackage{hyperref}
\usepackage{times}  
\usepackage{mathpazo}
\usepackage{amssymb,amsmath,amsthm}
\usepackage{epsfig}
\usepackage{enumerate}
\usepackage{algorithm}
\usepackage{algpseudocode}

\newtheorem{theorem}{Theorem}[section]
\newtheorem{proposition}[theorem]{Proposition}
\newtheorem{definition}[theorem]{Definition}

\newtheorem{claim}[theorem]{Claim}
\newtheorem{lemma}[theorem]{Lemma}

\newtheorem{corollary}[theorem]{Corollary}

\newtheorem{remark}[theorem]{Remark}

\newcommand{\qedsymb}{\hfill{\rule{2mm}{2mm}}}
\renewenvironment{proof}[1][]{\begin{trivlist}
\item[\hspace{\labelsep}{\bf\noindent Proof#1:\/}] }{\qedsymb\end{trivlist}}

\def\calF{{\cal F}}
\def\calE{{\cal E}}

\def\calM{{\cal M}}

\def\Q{\mathbb{Q}}

\newcommand\Prob[2]{{\Pr_{#1}\left[ {#2} \right]}}

\newcommand{\NP}{\mathsf{NP}}

\newcommand{\eps}{\epsilon}
\renewcommand{\epsilon}{\varepsilon}

\newcommand{\poly}{\mathop{\mathrm{poly}}}

\newcommand{\LEAF}{\textsc{Leaf}}
\newcommand{\SchrijverP}{\textsc{Schrijver}}
\newcommand{\KneserP}{\textsc{Kneser}}
\newcommand{\Agree}{\textsc{Agreeable-Set}}

\newcommand{\PPA}{\mathsf{PPA}}
\newcommand{\TFNP}{\mathsf{TFNP}}

\begin{document}

\title{{\bf Fixed-Parameter Algorithms for the Kneser and Schrijver Problems}\footnote{This paper subsumes the results of the two conference papers~\cite{Haviv22a,Haviv22b}.}}

\author{
Ishay Haviv\thanks{School of Computer Science, The Academic College of Tel Aviv-Yaffo, Tel Aviv 61083, Israel. Research supported in part by the Israel Science Foundation (grant No.~1218/20).
}
}

\date{}

\maketitle

\begin{abstract}
The Kneser graph $K(n,k)$ is defined for integers $n$ and $k$ with $n \geq 2k$ as the graph whose vertices are all the $k$-subsets of $[n]=\{1,2,\ldots,n\}$ where two such sets are adjacent if they are disjoint.
The Schrijver graph $S(n,k)$ is defined as the subgraph of $K(n,k)$ induced by the collection of all $k$-subsets of $[n]$ that do not include two consecutive elements modulo $n$.
It is known that the chromatic number of both $K(n,k)$ and $S(n,k)$ is $n-2k+2$.

In the computational $\KneserP$ and $\SchrijverP$ problems, we are given an access to a coloring with $n-2k+1$ colors of the vertices of $K(n,k)$ and $S(n,k)$ respectively, and the goal is to find a monochromatic edge.
We prove that the problems admit randomized algorithms with running time $n^{O(1)} \cdot k^{O(k)}$, hence they are fixed-parameter tractable with respect to the parameter $k$. The analysis involves structural results on intersecting families and on induced subgraphs of Kneser and Schrijver graphs.

We also study the $\Agree$ problem of assigning a small subset of a set of $m$ items to a group of $\ell$ agents, so that all agents value the subset at least as much as its complement. As an application of our algorithm for the $\KneserP$ problem, we obtain a randomized polynomial-time algorithm for the $\Agree$ problem for instances with $\ell \geq m - O(\frac{\log m}{\log \log m})$. We further show that the $\Agree$ problem is at least as hard as a variant of the $\KneserP$ problem with an extended access to the input coloring.
\end{abstract}

\newpage

\section{Introduction}

The Kneser graph $K(n,k)$ is defined for integers $n$ and $k$ with $n \geq 2k$ as the graph whose vertices are all the $k$-subsets of $[n] = \{1,2,\ldots,n\}$ where two such sets are adjacent if they are disjoint. In 1955, Kneser~\cite{Kneser55} observed that there exists a proper coloring of the vertices of $K(n,k)$ with $n-2k+2$ colors and conjectured that fewer colors do not suffice, that is, that its chromatic number satisfies $\chi(K(n,k)) = n-2k+2$.
The conjecture was proved more than two decades later by Lov\'asz~\cite{LovaszKneser} as a surprising application of the Borsuk--Ulam theorem from algebraic topology~\cite{Borsuk33}. Following this result, topological methods have become a common and powerful tool in combinatorics, discrete geometry, and theoretical computer science (see, e.g.,~\cite{MatousekBook}).
Several alternative proofs of Kneser's conjecture were provided in the literature over the years (see, e.g.,~\cite{MatousekZ04}), and despite the combinatorial nature of the conjecture, all of them essentially rely on the topological Borsuk--Ulam theorem. One exception is a proof by Matou{\v{s}}ek~\cite{Matousek04}, which is presented in a combinatorial form, but is yet inspired by a discrete variant of the Borsuk--Ulam theorem known as Tucker's lemma.

The Schrijver graph $S(n,k)$ is defined as the subgraph of $K(n,k)$ induced by the collection of all $k$-subsets of $[n]$ that do not include two consecutive elements modulo $n$ (i.e., the $k$-subsets $A \subseteq [n]$ such that if $i \in A$ then $i+1 \notin A$, and if $n \in A$ then $1 \notin A$).
Schrijver proved in~\cite{SchrijverKneser78}, strengthening Lov\'asz's result, that the chromatic number of $S(n,k)$ is equal to that of $K(n,k)$.
His proof technique relies on a proof of Kneser's conjecture due to B{\'{a}}r{\'{a}}ny~\cite{Barany78}, which was obtained soon after the one of Lov\'asz and combined the topological Borsuk--Ulam theorem with a lemma of Gale~\cite{Gale56}.
It was further proved in~\cite{SchrijverKneser78} that $S(n,k)$ is vertex-critical, that is, the chromatic number of any proper induced subgraph of $S(n,k)$ is strictly smaller than that of $S(n,k)$.

In the computational $\KneserP$ and $\SchrijverP$ problems, we are given an access to a coloring with $n-2k+1$ colors of the vertices of $K(n,k)$ and $S(n,k)$ respectively, and the goal is to find a monochromatic edge, i.e., two vertices with the same color that correspond to disjoint sets.
Since the number of colors used by the input coloring is strictly smaller than the chromatic number of the graph~\cite{LovaszKneser,SchrijverKneser78}, it follows that every instance of these problems has a solution.
However, the topological argument behind the lower bound on the chromatic number is not constructive, in the sense that it does not suggest an efficient algorithm for finding a monochromatic edge.
By an efficient algorithm we mean that its running time is polynomial in $n$, whereas the number of vertices might be exponentially larger.
Hence, it is natural to assume that the input coloring is given as an access to an oracle that given a vertex of the graph returns its color.
The input can also be given by some succinct representation, e.g., a Boolean circuit that computes the color of any given vertex.

In recent years, it has been shown that the complexity class $\PPA$ (Polynomial Parity Argument) perfectly captures the complexity of several total search problems for which the existence of the solution relies on the Borsuk--Ulam theorem.
This complexity class belongs to a family of classes that were introduced in 1994 by Papadimitriou~\cite{Papa94} in the attempt to characterize the mathematical arguments that lie behind the existence of solutions to search problems of $\TFNP$.
The complexity class $\TFNP$, introduced in~\cite{MegiddoP91}, is the class of total search problems in $\NP$, namely, the search problems in which a solution is guaranteed to exist and can be verified in polynomial time.
Papadimitriou has introduced in~\cite{Papa94} several subclasses of $\TFNP$, each of which consists of the total search problems that can be efficiently reduced to a problem that represents some mathematical argument.
One of those subclasses was $\PPA$ that corresponds to the fact that every (undirected) graph with maximum degree $2$ that has a vertex of degree $1$ must have another degree $1$ vertex. Hence, $\PPA$ is the class of all problems in $\TFNP$ that can be efficiently reduced to the $\LEAF$ problem, in which given a succinct representation of a graph with maximum degree $2$ and given a vertex of degree $1$ in the graph, the goal is to find another such vertex.

A prominent example of a $\PPA$-complete problem whose totality is related to the Borsuk--Ulam theorem is the one associated with the Consensus Halving theorem~\cite{HR65,SimmonsS03}.
The $\PPA$-completeness of the problem was proved for an inverse-polynomial precision parameter by Filos-Ratsikas and Goldberg~\cite{FG18,FG19}, and this was improved to a constant precision parameter in a recent work of Deligkas, Fearnley, Hollender, and Melissourgos~\cite{DeligkasFHM22}.
The hardness of the Consensus Halving problem was used to derive the $\PPA$-completeness of several other problems.
This includes the Splitting Necklace problem with two thieves, the Discrete Sandwich problem~\cite{FG18,FG19}, the Fair Independent Set in Cycle problem, and the aforementioned $\SchrijverP$ problem (with the input coloring given as a Boolean circuit)~\cite{Haviv22-FISC}.
As for the $\KneserP$ problem, the question of whether it is $\PPA$-complete was proposed by Deng, Feng, and Kulkarni~\cite{DengFK17} and is still open.
The question of determining the complexity of its extension to Kneser hypergraphs was recently raised by Filos-Ratsikas, Hollender, Sotiraki, and Zampetakis~\cite{Filos-RatsikasH21}.

The study of the $\KneserP$ problem has been recently motivated by its connection to a problem called $\Agree$ that was introduced by Manurangsi and Suksompong~\cite{ManurangsiS19} and further studied by Goldberg, Hollender, Igarashi, Manurangsi, and Suksompong~\cite{GoldbergHIMS20}.
This problem falls into the category of resource allocation problems, where one assigns items from a given collection $[m]$ to $\ell$ agents that have different preferences. The preferences are given by monotone utility functions that associate a non-negative value to each subset of $[m]$. In the $\Agree$ setting, the agents act as a group, and the goal is to collectively allocate a subset of items that is agreeable to all of them, in the sense that every agent likes it at least as much as it likes the complement set. The authors of~\cite{ManurangsiS19} proved that for every $\ell$ agents with monotone utility functions defined on the subsets of $[m]$, there exists a subset $S \subseteq [m]$ of size
\begin{eqnarray}\label{eq:agreeable_size}
|S| \leq \min \Big ( \Big \lfloor \frac{m+\ell}{2} \Big \rfloor,m \Big )
\end{eqnarray}
that is agreeable to all agents, and that this bound is tight in the worst case. They initiated the study of the $\Agree$ problem that given an oracle access to the utility functions of the $\ell$ agents, asks to find a subset $S \subseteq [m]$ that satisfies the worst-case bound given in~\eqref{eq:agreeable_size} and that is agreeable to all agents (see Theorem~\ref{thm:MS} and Definition~\ref{def:Agree}).
Note that for instances with $\ell \geq m$, the collection $S=[m]$ satisfies~\eqref{eq:agreeable_size} and thus forms a solution.

Interestingly, the proof given in~\cite{ManurangsiS19} for the existence of a solution to every instance the $\Agree$ problem relies on the chromatic number of Kneser graphs.
In fact, the proof implicitly shows that the $\Agree$ problem is efficiently reducible to the $\KneserP$ problem, hence the $\KneserP$ problem is at least as hard as the $\Agree$ problem.
An alternative existence proof, based on the Consensus Halving theorem~\cite{SimmonsS03}, was given by the authors of~\cite{GoldbergHIMS20}. Their approach was applied there to show that the $\Agree$ problem can be solved in polynomial time on instances with additive utility functions and on instances in which the number $\ell$ of agents is a fixed constant. The complexity of the problem in the general case is still open.

\subsection{Our Contribution}

Our main result concerns the parameterized complexity of the $\KneserP$ and $\SchrijverP$ problems.
In the area of parameterized complexity, a decision problem whose instances involve a parameter $k$ is said to be fixed-parameter tractable with respect to $k$ if it admits an algorithm whose running time is bounded by a polynomial in the input size multiplied by an arbitrary function of $k$ (see, e.g.,~\cite{CyganFKLMPPS15}).
We adopt the notion of fixed-parameter tractability to our setting, where the studied problems are total search problems rather than decision problems, and where the inputs are given as an oracle access.
We say that an algorithm for the $\KneserP$ and $\SchrijverP$ problems is fixed-parameter with respect to the parameter $k$ if its running time on an input coloring of, respectively, $K(n,k)$ and $S(n,k)$ is bounded by $n^{O(1)} \cdot f(k)$ for some function $f$.
The following theorem shows that the $\SchrijverP$ problem is fixed-parameter tractable with respect to $k$.

\begin{theorem}\label{thm:AlgoKneserNew}
There exists a randomized algorithm that given integers $n$ and $k$ with $n \geq 2k$ and an oracle access to a coloring of the vertices of the Schrijver graph $S(n,k)$ with $n-2k+1$ colors, runs in time $n^{O(1)} \cdot k^{O(k)}$ and returns a monochromatic edge with high probability.
\end{theorem}

A few remarks about Theorem~\ref{thm:AlgoKneserNew} are in order here.
\begin{itemize}
  \item The study of the fixed-parameter tractability of the $\SchrijverP$ problem is motivated by its $\PPA$-hardness proved in~\cite{Haviv22-FISC}.
  \item The $\SchrijverP$ problem is at least as hard as the $\KneserP$ problem. Indeed, the Schrijver graph $S(n,k)$ is an induced subgraph of the Kneser graph $K(n,k)$ with the same chromatic number. Therefore, the $\KneserP$ problem can be solved by applying an algorithm for the $\SchrijverP$ problem to the restriction of a coloring of a Kneser graph to its Schrijver subgraph. This implies that Theorem~\ref{thm:AlgoKneserNew} in particular implies that the $\KneserP$ problem is also fixed-parameter tractable with respect to the parameter $k$.
  \item As mentioned earlier, the Schrijver graph $S(n,k)$ was shown in~\cite{SchrijverKneser78} to be vertex-critical. It follows that for every vertex $A$ of the graph $S(n,k)$, there exists a coloring of its vertices with $n-2k+1$ colors, for which only edges that are incident with $A$ are monochromatic. An algorithm for the $\SchrijverP$ problem, while running on such an input coloring, must be able to find an edge that is incident with this specified vertex $A$. Nevertheless, the algorithm given in Theorem~\ref{thm:AlgoKneserNew} manages to do so in running time much smaller than the number of vertices, provided that $n$ is sufficiently larger than $k$.
  \item Borrowing the terminology of the area of parameterized complexity, our algorithm for the $\SchrijverP$ problem can be viewed as a randomized polynomial Turing kernelization algorithm for the problem (see, e.g.,~\cite[Chapter~22]{KernelBook19}). Namely, the problem of finding a monochromatic edge in a Schrijver graph $S(n,k)$ can essentially be reduced by a randomized efficient algorithm to finding a monochromatic edge in a Schrijver graph $S(n',k)$ for $n' = O(k^4)$. For the $\KneserP$ problem, we provide a tighter analysis, allowing us to efficiently reduce the problem of finding a monochromatic edge in a Kneser graph $K(n,k)$ to finding such an edge in $K(n',k)$ for $n' = O(k^3)$. See Section~\ref{sec:kernel} for the precise details.
  \item The analysis of the algorithm given in Theorem~\ref{thm:AlgoKneserNew} relies on structural properties of induced subgraphs of Schrijver graphs (see Section~\ref{sec:properties}). The proofs involve an idea recently applied by Frankl and Kupavskii~\cite{FrankK20} in the study of maximal degrees in induced subgraphs of Kneser graphs as well as some tools developed by Istrate, Bonchis, and Craciun~\cite{IstrateBC21} in the context of Frege propositional proof systems. An overview of the proof of Theorem~\ref{thm:AlgoKneserNew} is given in Section~\ref{sec:overview}.
\end{itemize}

Our next result provides a simple deterministic algorithm for the $\KneserP$ problem that is particularly useful for Kneser graphs $K(n,k)$ with $k$ close to $n/2$.
Its analysis is based on the chromatic number of Schrijver graphs~\cite{SchrijverKneser78}.
\begin{theorem}\label{thm:KneserAlgSch}
There exists an algorithm that given integers $n$ and $k$ with $n \geq 2k$ and an oracle access to a coloring $c: \binom{[n]}{k} \rightarrow [n-2k+1]$ of the vertices of the Kneser graph $K(n,k)$, returns a monochromatic edge in running time polynomial in $\frac{n}{k} \cdot \binom{n-k-1}{k-1} \leq n^{\min(k,n-2k+1)}$.
\end{theorem}

We proceed by presenting our results on the $\Agree$ problem.
First, as applications of Theorems~\ref{thm:AlgoKneserNew} and~\ref{thm:KneserAlgSch}, using the relation from~\cite{ManurangsiS19} between the $\Agree$ and $\KneserP$ problems, we obtain the following algorithmic results.

\begin{theorem}\label{thm:AlgoIntro1}
There exists a randomized algorithm for the $\Agree$ problem that given an oracle access to an instance with $m$ items and $\ell$ agents ($\ell<m$), runs in time $m^{O(1)} \cdot k^{O(k)}$ for $k = \lceil \frac{m-\ell}{2} \rceil$ and returns a solution with high probability.
\end{theorem}

\begin{theorem}\label{thm:AlgoIntro2}
There exists an algorithm for the $\Agree$ problem that given an oracle access to an instance with $m$ items and $\ell$ agents ($\ell<m$), returns a solution in running time polynomial in $\frac{m}{k} \cdot \binom{m-k-1}{k-1} \leq m^{\min(k,m-2k+1)}$ for $k = \lceil \frac{m-\ell}{2} \rceil$.
\end{theorem}
\noindent
We apply Theorems~\ref{thm:AlgoIntro1} and~\ref{thm:AlgoIntro2} to show that the $\Agree$ problem can be solved in polynomial time for certain families of instances.
By Theorem~\ref{thm:AlgoIntro1}, we obtain a randomized efficient algorithm for instances in which the number of agents $\ell$ is not much smaller than the number of items $m$, namely, for $\ell \geq m - O(\frac{\log m}{\log \log m})$ (see Corollary~\ref{cor:m<=l+log/loglog}).
By Theorem~\ref{thm:AlgoIntro2}, we obtain an efficient algorithm for instances with a constant number of agents (see Corollary~\ref{cor:constant_l}), providing an alternative proof for a result of~\cite{GoldbergHIMS20}.

We finally explore the relations between the $\Agree$ and $\KneserP$ problems.
As already mentioned, there exists an efficient reduction from the $\Agree$ problem to the $\KneserP$ problem~\cite{ManurangsiS19}.
Here we provide a reduction in the opposite direction.
However, for the reduction to be efficient we reduce from a variant of the $\KneserP$ problem with an extended type of queries, which we call {\em subset queries} and define as follows.
For an input coloring $c : \binom{[n]}{k} \rightarrow [n-2k+1]$ of a Kneser graph $K(n,k)$, a subset query is a pair $(i,B)$ of a color $i \in [n-2k+1]$ and a set $B \subseteq [n]$, and the answer on the query $(i,B)$ determines whether $B$ contains a vertex colored $i$, that is, whether there exists a $k$-subset $A \subseteq B$ satisfying $c(A)=i$.
We prove the following result (see Section~\ref{sec:models} for the computational input model of the problems).

\begin{theorem}\label{thm:KneserVSAgree}
There exists a polynomial-time reduction from the $\KneserP$ problem with subset queries to the $\Agree$ problem.
\end{theorem}

We believe that the approach taken in the present paper, of studying problems in $\TFNP$ from the parameterized complexity perspective, warrants further attention in future research (see also, e.g.,~\cite{FHSZ20,BSW23}). This approach might lead to a better understanding of the complexity of such problems by identifying the regimes of the involved parameters for which efficient algorithms exist. From a mathematical point of view, this might help us to gain insights on the regimes of parameters for which the mathematical arguments used in the existence proof of the solutions are essential.

\subsection{Proof Overview of Theorem~\ref{thm:AlgoKneserNew}}\label{sec:overview}

We present here the main ideas of our fixed-parameter algorithm for the $\SchrijverP$ problem.
For convenience, we start by describing a fixed-parameter algorithm for the $\KneserP$ problem and then discuss the modification in the algorithm and in its analysis needed for the $\SchrijverP$ problem.

Suppose that we are given, for $n \geq 2k$, an oracle access to a coloring $c: \binom{[n]}{k} \rightarrow [n-2k+1]$ of the vertices of the Kneser graph $K(n,k)$ with $n-2k+1$ colors. As mentioned before, since the chromatic number of $K(n,k)$ is $n-2k+2$~\cite{LovaszKneser}, the coloring $c$ must have a monochromatic edge. Such an edge can clearly be found by an algorithm that queries the oracle for the colors of all the vertices.
However, the running time of such an algorithm is polynomial in $\binom{n}{k}$, so it is not fixed-parameter with respect to the parameter $k$.

A natural attempt to improve on this running time is to consider a randomized algorithm that picks uniformly and independently at random polynomially many vertices of $K(n,k)$ and checks if any two of them form a monochromatic edge. However, it is not difficult to see that there exist colorings of $K(n,k)$ with $n-2k+1$ colors for which a small fraction of the vertices are involved in all the monochromatic edges, implying that the success probability of this randomized algorithm is negligible on them.
To see this, consider the canonical coloring of $K(n,k)$, in which for every $i \in [n-2k+1]$ the vertices $A \in \binom{[n]}{k}$ with $\min(A)=i$ are colored $i$, and the remaining vertices, those contained in $[n] \setminus [n-2k+1]$, are colored $n-2k+2$. By recoloring the vertices of the last color class with arbitrary colors from $[n-2k+1]$, we get a coloring of $K(n,k)$ with $n-2k+1$ colors such that every monochromatic edge has an endpoint in a collection of $\binom{2k-1}{k}$ vertices.
This implies that the probability that two random vertices chosen uniformly and independently from $\binom{[n]}{k}$ form a monochromatic edge may go to zero faster than any inverse polynomial in $n$.

While the above coloring shows that all the monochromatic edges can involve vertices from a small set, one may notice that this coloring is very well structured, in the sense that each color class is quite close to a trivial intersecting family, i.e., an intersecting family all of whose members share a common element.
This is definitely not a coincidence. It is known that large intersecting families of $k$-subsets of $[n]$ are `essentially' contained in trivial intersecting families.
Indeed, the classical Erd{\"{o}}s--Ko--Rado theorem~\cite{EKR61} asserts that the largest size of an intersecting family of $k$-subsets of $[n]$ is $\binom{n-1}{k-1}$, attained by, and only by, the $n$ largest trivial families (for $n>2k$). Moreover, Hilton and Milner~\cite{HM67} proved a stability result for the Erd{\"{o}}s--Ko--Rado theorem, saying that if an intersecting family of sets from $\binom{[n]}{k}$ is not trivial then its size cannot exceed $\binom{n-1}{k-1}-\binom{n-k-1}{k-1}+1$, which is much smaller than the largest possible size of an intersecting family when $n$ is sufficiently larger than $k$ (see Theorem~\ref{thm:HM} and Remark~\ref{remark:HM}). More recently, it was shown by Dinur and Friedgut~\cite{DinurF09} that every intersecting family can be made trivial by removing not more than $\tilde{c} \cdot \binom{n-2}{k-2}$ of its members for a constant $\tilde{c}$ (see~\cite{Kupavskii18} for an exact $\binom{n-3}{k-2}$ bound on the number of sets that should be removed, provided that $n \geq \tilde{c} \cdot k$ for a constant $\tilde{c}$).
Hence, our strategy for finding a monochromatic edge in $K(n,k)$ is to learn the structure of the large color classes which are close to being intersecting.
We use random samples from the vertex set of the graph in order to identify the common elements of the trivial families that `essentially' contain these color classes.
Roughly speaking, this allows us to repeatedly reduce the size of the ground set $[n]$ of the given Kneser graph and to obtain a small subgraph, whose size depends only on $k$, that is expected to contain a monochromatic edge. Then, such an edge can simply be found by querying the oracle for the colors of all the vertices of this subgraph.

With the above idea in mind, let us consider an algorithm that starts by selecting uniformly and independently at random polynomially many vertices of $K(n,k)$ and queries the oracle for their colors. Among the $n-2k+1$ color classes, there must be a non-negligible one that includes at least $\frac{1}{n-2k+1} \geq \frac{1}{n}$ fraction of the vertices. It can be shown, using the Chernoff--Hoeffding bound, that the samples of the algorithm can be used to learn a color $i \in [n-2k+1]$ of a quite large color class. Moreover, the samples can be used to identify an element $j \in [n]$ that is particularly popular on the vertices colored $i$ (say, that belongs to a constant fraction of them) in case that such an element exists. Now, if the coloring satisfies that (a) there exists an element $j$ that is popular on the vertices colored $i$ and, moreover, (b) this element $j$ belongs to {\em all} the vertices that are colored $i$, then it suffices to focus on the subgraph of $K(n,k)$ induced by the vertices of $\binom{[n] \setminus \{j\}}{k}$. Indeed, condition (b) implies that the restriction of the given coloring to this subgraph uses at most $n-2k$ colors.
Since this subgraph is isomorphic to $K(n-1,k)$, its chromatic number is $n-2k+1$, hence it has a monochromatic edge. By repeatedly applying this procedure, assuming that the conditions (a) and (b) hold in all iterations, we can eliminate elements from the ground set $[n]$ and obtain smaller and smaller Kneser graphs that still have a monochromatic edge. When the ground set becomes sufficiently small, one can go over all the remaining vertices and efficiently find the required edge.
We address now the case where at least one of the conditions (a) and (b) does not hold.

For condition (a), suppose that in some iteration the algorithm identifies a color $i$ that appears on a significant fraction of vertices, but no element of $[n]$ is popular on these vertices.
In this case, one might expect the color class of $i$ to be so far from being intersecting, so that the polynomially many samples would include with high probability a monochromatic edge of vertices colored $i$.
It can be observed that the aforementioned stability results of the Erd{\"{o}}s--Ko--Rado theorem yield that a non-negligible fraction of the vertices of a large color class with no popular element lie on monochromatic edges.
However, in order to catch a monochromatic edge with high probability we need here the stronger requirement, saying that a non-negligible fraction of the {\em pairs} of vertices from this color class form monochromatic edges. This indeed holds for color classes of $K(n,k)$ that include at least, say, $\frac{1}{n}$ fraction of the vertices, provided that $n \geq \poly(k)$. The proof given in~\cite[Lemma~3.1]{Haviv22a} uses the Hilton--Milner theorem~\cite{HM67} and borrows an idea of~\cite{FrankK20} (see also Lemma~\ref{lemma:k^3}).

For condition (b), suppose that in some iteration the algorithm identifies a color $i$ of a large color class and an element $j$ that is popular in its sets but does not belong to all of them. In this case, we can no longer ensure that the final Kneser graph obtained after all iterations has a monochromatic edge, as some of its vertices might be colored $i$.
To handle this situation, we show that every vertex $A \in \binom{[n]}{k}$ with $j \notin A$ has a lot of neighbors colored $i$.
Hence, if the algorithm finds at its final step, or even earlier, a vertex $A$ colored $i$ satisfying $j \notin A$, where $i$ and $j$ are the color and element chosen by the algorithm in a previous iteration, then it goes back to the ground set of this iteration and finds using additional polynomially many random vertices from it a neighbor of $A$ colored $i$, and thus a monochromatic edge.

To summarize, our algorithm for the $\KneserP$ problem repeatedly calls an algorithm, which we refer to as the `element elimination' algorithm, that uses polynomially many random vertices to identify a color $i$ of a large color class. If no element of $[n]$ is popular on this color class then the random samples provide a monochromatic edge with high probability and we are done. Otherwise, the algorithm finds such a popular $j \in [n]$ and focuses on the subgraph obtained by eliminating $j$ from the ground set. This is done as long as the size of the ground set is larger than some polynomial in $k$. After all iterations, the remaining vertices induce a Kneser graph $K(n',k)$ for $n' \leq \poly(k)$, and the algorithm queries for the colors of all of its vertices in time polynomial in $\binom{n'}{k} \leq k^{O(k)}$. If the colors that were chosen through the iterations of the `element elimination' algorithm do not appear in this subgraph, then it must contain a monochromatic edge which can be found by an exhaustive search. Otherwise, this search gives us a vertex $A$ satisfying $c(A) = i$ and $j \notin A$ for a color $i$ associated with an element $j$ by one of the calls to the `element elimination' algorithm. As explained above, given such an $A$ it is possible to efficiently find with high probability a vertex that forms with $A$ a monochromatic edge.
This gives us a randomized algorithm with running time $n^{O(1)} \cdot k^{O(k)}$ that finds a monochromatic edge with high probability.

The above approach for the $\KneserP$ problem can be modified to obtain our algorithm for the $\SchrijverP$ problem given in Theorem~\ref{thm:AlgoKneserNew}.
This requires a refined analysis of the `element elimination' algorithm for Schrijver graphs.
Consider first the case where the input coloring has a large color class that does not have a popular element in its members.
For this case we prove that the selected random vertices include a monochromatic edge with high probability.
In contrast to the analysis used in~\cite{Haviv22a} for Kneser graphs, here we cannot apply the Hilton--Milner theorem~\cite{HM67} that deals with intersecting families of general $k$-subsets of $[n]$. We overcome this issue using a Hilton--Milner-type result for stable sets, i.e., for vertices of the Schrijver graph, borrowing ideas that were applied by Istrate, Bonchis, and Craciun~\cite{IstrateBC21} in the context of Frege propositional proof systems (see Lemma~\ref{lemma:HMstable}). Note that this can be interpreted as an approximate stability result for the analogue of the Erd{\"{o}}s--Ko--Rado theorem for stable sets that was proved in 2003 by Talbot~\cite{Talbot03}.
The Hilton--Milner-type result is combined with an idea of~\cite{FrankK20} to prove that if the vertices of a large color class do not have a popular element, then a pair of vertices chosen uniformly at random from $S(n,k)$ forms a monochromatic edge with a non-negligible probability (see Lemma~\ref{lemma:at_most_gamma} and Corollary~\ref{cor:at_most_gamma}). Hence, picking a polynomial number of them suffices to catch such an edge.

Consider next the case where every large color class of the input coloring has a popular element.
Here, the `element elimination' algorithm identifies with high probability a color $i$ of a large color class and an element $j$ that is popular on its vertices.
If all the vertices colored $i$ include $j$ then we can safely look for a monochromatic edge in the subgraph of $S(n,k)$ induced by the cyclic ordering of $[n]$ without the element $j$, as this means that the size of the ground set and the number of colors are both reduced by $1$. However, for the scenario where the color class of $i$ involves vertices $A$ that do not include $j$, we prove that such an $A$ is disjoint from a random set from the color class of $i$ with a non-negligible probability. Note that the analysis in this case again employs the ideas applied by Istrate et al.~\cite{IstrateBC21} (see Lemma~\ref{lemma:at_least_gamma} and Corollary~\ref{cor:at_least_gamma}).
Then, when such a set $A$ is found by the algorithm, we can go back to the subgraph of the run of the `element elimination' algorithm that identified the color of $A$ and find a neighbor of $A$ from this color class using random samples.
The full description of the algorithm and its analysis are presented in Section~\ref{sec:algo}.

\subsection{Outline}
The rest of the paper is organized as follows.
In Section~\ref{sec:pre}, we gather several definitions and results that will be used throughout the paper.
In Section~\ref{sec:properties}, we present and prove several structural results on induced subgraphs of Kneser and Schrijver graphs.
In Section~\ref{sec:algo}, we present and analyze our randomized fixed-parameter algorithm for the $\SchrijverP$ problem and prove Theorem~\ref{thm:AlgoKneserNew}.
In Section~\ref{sec:AlgoSchr}, we present a simple deterministic algorithm for the $\KneserP$ problem and prove Theorem~\ref{thm:KneserAlgSch}.
In Section~\ref{sec:agree}, we study the $\Agree$ problem. We prove there Theorems~\ref{thm:AlgoIntro1} and~\ref{thm:AlgoIntro2} and derive efficient algorithms for certain families of instances.
We also show that the $\Agree$ problem is at least as hard as the $\KneserP$ problem with subset queries, confirming Theorem~\ref{thm:KneserVSAgree}.

\section{Preliminaries}\label{sec:pre}

\subsection{Computational Models}\label{sec:models}

In this work we consider total search problems whose inputs involve functions that are defined on domains of size exponential in the parameters of the problems.
For example, the input of the $\KneserP$ problem is a coloring $c: \binom{[n]}{k} \rightarrow [n-2k+1]$ of the vertices of the Kneser graph $K(n,k)$ for $n \geq 2k$.
For such problems, one has to specify how the input is given. We consider the following two input models.
\begin{itemize}
  \item In the {\em black-box input model}, an input function is given as an oracle access, so that an algorithm can query the oracle for the value of the function on any element of its domain. This input model is used in the current work to present our algorithmic results, reflecting the fact that the algorithms do not rely on the representation of the input function.
  \item In the {\em white-box input model}, an input function is given by a succinct representation that can be used to efficiently determine the values of the function, e.g., a Boolean circuit or an efficient Turing machine. This input model is appropriate to study the computational complexity of problems, and in particular, to show membership and hardness results with respect to the complexity class $\PPA$.
\end{itemize}

Reductions form a useful tool to show relations between problems.
Let $P_1$ and $P_2$ be total search problems.
We say that $P_1$ is (polynomial-time) reducible to $P_2$ if there exist (polynomial-time) computable functions $f,g$ such that $f$ maps any input $x$ of $P_1$ to an input $f(x)$ of $P_2$, and $g$ maps any pair $(x,y)$ of an input $x$ of $P_1$ and a solution $y$ of $f(x)$ with respect to $P_2$ to a solution of $x$ with respect to $P_1$.
For problems $P_1$ and $P_2$ in the black-box input model, one has to use the notion of {\em black-box reductions}. A (polynomial-time) black-box reduction satisfies that the oracle access needed for the input $f(x)$ of $P_2$ can be simulated by a (polynomial-time) procedure that has an oracle access to the input $x$. In addition, the solution $g(x,y)$ of $x$ in $P_1$ can be computed (in polynomial time) given the solution $y$ of $f(x)$ and the oracle access to the input $x$.
In the current work we will use black-box reductions to obtain algorithmic results for problems in the black-box input model.
For more details on these concepts, we refer the reader to~\cite[Section~2.2]{BeameCEIP98} (see also~\cite[Sections~2 and~4]{DeligkasFH21} for related discussions).

\subsection{Kneser and Schrijver Graphs}\label{sec:KneSchr}

Consider the following definition.
\begin{definition}\label{def:K(F)}
For a family $\calF$ of non-empty sets, let $K(\calF)$ denote the graph on the vertex set $\calF$ in which two vertices are adjacent if they represent disjoint sets.
\end{definition}
\noindent
For a set $X$ and an integer $k$, let $\binom{X}{k}$ denote the family of all $k$-subsets of $X$.
Note that the {\em Kneser graph} $K(n,k)$ can be defined for integers $n$ and $k$ with $n \geq 2k$ as the graph $K(\binom{[n]}{k})$.

A set $A \subseteq [n]$ is said to be {\em stable} if it does not include two consecutive elements modulo $n$, that is, it forms an independent set in the $n$-vertex cycle with the numbering from $1$ to $n$ along the cycle. For integers $n$ and $k$ with $n \geq 2k$, let $\binom{[n]}{k}_{\mathrm{stab}}$ denote the collection of all stable $k$-subsets of $[n]$. The {\em Schrijver graph} $S(n,k)$ is defined as the graph $K(\binom{[n]}{k}_{\mathrm{stab}})$. Equivalently, it is the subgraph of $K(n,k)$ induced by the vertex set $\binom{[n]}{k}_{\mathrm{stab}}$.

For a set $X \subseteq [n]$ consider the natural cyclic ordering of the elements of $X$ induced by that of $[n]$, and let $\binom{X}{k}_{\mathrm{stab}}$ denote the collection of all $k$-subsets of $X$ that do not include two consecutive elements according to this ordering. More formally, letting $j_1 < j_2 < \cdots < j_{|X|}$ denote the elements of $X$, $\binom{X}{k}_{\mathrm{stab}}$ stands for the collection of all independent sets of size $k$ in the cycle on the vertex set $X$ with the numbering $j_1, \ldots, j_{|X|}$ along the cycle. Note that $\binom{X}{k}_{\mathrm{stab}} \subseteq \binom{X'}{k}_{\mathrm{stab}}$ whenever $X \subseteq X' \subseteq [n]$.
Note further that the graph $K(\binom{X}{k}_{\mathrm{stab}})$ is isomorphic to the Schrijver graph $S(|X|,k)$.

The chromatic numbers of the graphs $K(n,k)$ and $S(n,k)$ were determined, respectively, by Lov\'asz~\cite{LovaszKneser} and by Schrijver~\cite{SchrijverKneser78} as follows.

\begin{theorem}[\cite{LovaszKneser,SchrijverKneser78}]\label{thm:KneserS}
For all integers $n \geq 2k$, $\chi(K(n,k)) = \chi(S(n,k)) = n-2k+2$.
\end{theorem}


The computational search problems associated with the Kneser and Schrijver graphs are defined as follows.
\begin{definition}\label{def:KneserSProblems}
In the computational $\KneserP$ problem, the input is a coloring $c: \binom{[n]}{k} \rightarrow [n-2k+1]$ of the vertices of the Kneser graph $K(n,k)$ with $n-2k+1$ colors for integers $n$ and $k$ with $n \geq 2k$, and the goal is to find a monochromatic edge, i.e., $A,B \in \binom{[n]}{k}$ satisfying $A \cap B = \emptyset$ and $c(A) = c(B)$.
In the computational $\SchrijverP$ problem, the input is a coloring $c: \binom{[n]}{k}_{\mathrm{stab}} \rightarrow [n-2k+1]$ of the vertices of the Schrijver graph $S(n,k)$ with $n-2k+1$ colors for integers $n$ and $k$ with $n \geq 2k$, and the goal is to find a monochromatic edge.
In the black-box input model, the input coloring is given as an oracle access that for a vertex $A$ returns its color $c(A)$.
In the white-box input model, the input coloring is given by a Boolean circuit that for a vertex $A$ computes its color $c(A)$.
\end{definition}
\noindent
The existence of a solution for every instance of the $\KneserP$ and $\SchrijverP$ problems follows from Theorem~\ref{thm:KneserS}.

\subsection{Intersecting Families}

For integers $n$ and $k$ with $n \geq 2k$, let $\calF \subseteq \binom{[n]}{k}$ be a family of $k$-subsets of $[n]$.
We call $\calF$ {\em intersecting} if for every two sets $F_1, F_2 \in \calF$ it holds that $F_1 \cap F_2 \neq \emptyset$.
Note that such a family forms an independent set in the graph $K(n,k)$.
The Erd{\"{o}}s--Ko--Rado theorem~\cite{EKR61} asserts that every intersecting family $\calF \subseteq \binom{[n]}{k}$ satisfies $|\calF| \leq \binom{n-1}{k-1}$.
This bound is tight and is attained, for each $i \in [n]$, by the family $\{ A \in \binom{[n]}{k} \mid i \in A \}$.
An intersecting family of sets is said to be {\em trivial} if its members share a common element.
Hilton and Milner~\cite{HM67} proved the following stability result for the Erd{\"{o}}s--Ko--Rado theorem, providing an upper bound on the size of any non-trivial intersecting family.
\begin{theorem}[Hilton--Milner Theorem~\cite{HM67}]\label{thm:HM}
For all integers $k \geq 2$ and $n \geq 2k$, every non-trivial intersecting family of $k$-subsets of $[n]$ has size at most $\binom{n-1}{k-1}-\binom{n-k-1}{k-1}+1$.
\end{theorem}
\begin{remark}\label{remark:HM}
The bound given in Theorem~\ref{thm:HM} is tight.
To see this, for an arbitrary $k$-subset $F$ of $[n]$ and for an arbitrary element $i \notin F$, consider the family $\calF = \{ A \in \binom{[n]}{k} \mid A \cap F \neq \emptyset,~i \in A\} \cup \{F\}$.
The family $\calF$ is intersecting and non-trivial, and its size coincides with the bound given in Theorem~\ref{thm:HM}. Note that $|\calF| \leq k \cdot \binom{n-2}{k-2}$, provided that $k \geq 3$.
\end{remark}

\subsection{Chernoff--Hoeffding Bound}
We need the following concentration result (see, e.g.,~\cite[Theorem~2.1]{McDiarmid98}).
\begin{theorem}[Chernoff--Hoeffding Bound]\label{thm:chernoff}
Let $0 < p < 1$, let $X_1,\ldots,X_m$ be $m$ independent binary random variables satisfying $\Prob{}{X_i=1}=p$ and $\Prob{}{X_i=0}=1-p$ for all $i$, and put $\overline{X} = \frac{1}{m} \cdot \sum_{i=1}^{m}{X_i}$. Then, for any $\mu \geq 0$,
\[ \Prob{}{|\overline{X}-p| \geq \mu} \leq 2 \cdot e^{-2m \mu^2}.\]
\end{theorem}

\section{Induced Subgraphs of Kneser and Schrijver Graphs}\label{sec:properties}

In this section, we provide a couple of lemmas on induced subgraphs of Schrijver graphs that will play a central role in the analysis of our fixed-parameter algorithm for the $\SchrijverP$ problem. We further discuss some properties of induced subgraphs of Kneser graphs.
We start with some preliminary claims related to counting stable sets.

\subsection{Counting Stable Sets}

We need the following simple claim.
\begin{claim}\label{claim:stable_k_n}
For integers $k \geq 1$ and $n \geq 2k$, the number of $k$-subsets of $[n]$ with no two consecutive elements (allowing both $1$ and $n$ to be in the subsets) is $\binom{n-k+1}{k}$.
\end{claim}
\begin{proof}
Identify the subsets of $[n]$ with their characteristic vectors in $\{0,1\}^n$, and in every such vector, interpret the zeros as balls and the ones as separations between bins. It follows that every $k$-subset of $[n]$ with no two consecutive elements corresponds to a partition of $n-k$ identical balls into $k+1$ bins, where all bins but the first and last are not empty. The number of those partitions is equal to the number of partitions of $n-2k+1$ identical balls into $k+1$ bins, which is $\binom{n-k+1}{k}$.
\end{proof}

Using Claim~\ref{claim:stable_k_n}, we derive the number of vertices in Schrijver graphs.

\begin{claim}\label{claim:|V(S(n,k))|}
For integers $k \geq 2$ and $n \geq 2k$, the number of vertices in $S(n,k)$ is $\frac{n}{k} \cdot \binom{n-k-1}{k-1}$.
\end{claim}

\begin{proof}
Recall that the vertex set $\binom{[n]}{k}_{\mathrm{stab}}$ of $S(n,k)$ is the collection of all $k$-subsets of $[n]$ with no two consecutive elements modulo $n$.
By Claim~\ref{claim:stable_k_n}, the number of $k$-subsets of $[n]$ with no two consecutive elements (allowing both $1$ and $n$ to be in the subsets) is $\binom{n-k+1}{k}$.
To determine the number of vertices in $S(n,k)$, we consider the following two cases.
The number of vertices that do not include $1$ nor $n$ is equal to the number of $k$-subsets of $[n-2]$ with no two consecutive elements and thus to $\binom{n-k-1}{k}$.
The number of vertices that include either $1$ or $n$ is twice the number of $(k-1)$-subsets of $[n-3]$ with no two consecutive elements, and the latter is equal to $\binom{n-k-1}{k-1}$. It thus follows that the number of vertices in $S(n,k)$ is
\begin{eqnarray*}
\binom{n-k-1}{k}+2 \cdot \binom{n-k-1}{k-1} = \frac{n}{k} \cdot \binom{n-k-1}{k-1},
\end{eqnarray*}
where the equality holds by a straightforward calculation.
This completes the proof.
\end{proof}

The following claim employs an argument of~\cite{IstrateBC21}.
We include its proof for completeness.

\begin{claim}\label{claim:stable_i_j}
For integers $k \geq 2$ and $n \geq 2k$ and for every two distinct integers $a,b \in [n]$, the number of stable $k$-subsets $F$ of $[n]$ satisfying $\{a,b\} \subseteq F$ is at most $\binom{n-k-2}{k-2}$.
\end{claim}

\begin{proof}
It is easy to see that the statement of the claim holds for $k=2$, hence we assume from now on that $k \geq 3$. Suppose without loss of generality that $a=1$. For an integer $b \in [n] \setminus \{1\}$, consider the stable $k$-subsets $F$ of $[n]$ satisfying $\{1,b\} \subseteq F$.

For $b \in \{2,n\}$, there are no such sets, hence the required bound trivially holds.

For $b \in \{3,n-1\}$, assume without loss of generality that $b=3$, and observe that the desired sets $F$ have the form $F' \cup \{1,3\}$, where $F'$ is a $(k-2)$-subset of $\{5, \ldots, n-1\}$ with no two consecutive elements.
In case that $n \geq 2k+1$, we have $n-5 \geq 2 \cdot (k-2)$, so using $k \geq 3$, we can apply Claim~\ref{claim:stable_k_n} to obtain that the number of sets is
\[\binom{(n-5)-(k-2)+1}{k-2} = \binom{n-k-2}{k-2},\]
and therefore the required bound holds.
For $n=2k$, notice that there is a single $F'$ as above, so the bound holds as well.

Similarly, for $b \in \{4,n-2\}$, assume without loss of generality that $b=4$, for which the desired sets $F$ have the form $F' \cup \{1,4\}$, where $F'$ is a $(k-2)$-subset of $\{6, \ldots, n-1\}$ with no two consecutive elements. As before, for $n \geq 2k+2$, we can apply Claim~\ref{claim:stable_k_n} to obtain that the number of sets is
\[\binom{(n-6)-(k-2)+1}{k-2} = \binom{n-k-3}{k-2} \leq \binom{n-k-2}{k-2}.\]
The bound also holds for $n \in \{2k,2k+1\}$, for which there is at most one $F'$ as above.

Otherwise, for $b \in \{5,\ldots,n-3\}$, the desired sets $F$ have the form $F' \cup \{1,b\}$, where $F'$ is a $(k-2)$-subset of $\{3, \ldots, b-2\} \cup \{b+2,\ldots,n-1\}$ with no two consecutive elements. By shifting the elements of $F'$ in the range $\{3, \ldots, b-2\}$ by $2$, it follows that the number of those sets $F'$ is equal to the number of $(k-2)$-subsets of $\{5, \ldots,n-1\}$ with no two consecutive elements and without the element $b+1$. As shown above (for $b=3$), the number of these subsets, even without the requirement not to include $b+1$, does not exceed the required bound, so we are done.
\end{proof}

\begin{remark}
It follows from the above proof that the bound in Claim~\ref{claim:stable_i_j} is tight in certain cases (e.g., for $a=1$, $b=3$).
\end{remark}

The following result stems from Claim~\ref{claim:stable_i_j} and can be viewed as an approximate variant of the Hilton--Milner theorem of~\cite{HM67} for stable sets (see~\cite{IstrateBC21}).
\begin{lemma}\label{lemma:HMstable}
For integers $k \geq 2$ and $n \geq 2k$, every non-trivial intersecting family $\calF$ of stable $k$-subsets of $[n]$ satisfies $|\calF| \leq k^2 \cdot \binom{n-k-2}{k-2}$.
\end{lemma}

\begin{proof}
Let $\calF$ be a non-trivial intersecting family of stable $k$-subsets of $[n]$.
Consider an arbitrary set $A = \{a_1, \ldots, a_k\}$ in $\calF$. Since $\calF$ is non-trivial, for every $t \in [k]$, there exists a set $B_t \in \calF$ satisfying $a_t \notin B_t$.
Since $\calF$ is intersecting, every set in $\calF$ intersects $A$, and therefore includes the element $a_t$ for some $t \in [k]$. Such a set further intersects $B_t$, hence it also includes some element $b \in B_t$ (which is different from $a_t$).
By Claim~\ref{claim:stable_i_j}, the number of stable $k$-subsets of $[n]$ that include both $a_t$ and $b$ does not exceed $\binom{n-k-2}{k-2}$.
Since there are at most $k^2$ ways to choose the elements $a_t$ and $b$, this implies that $|\calF| \leq k^2 \cdot \binom{n-k-2}{k-2}$, as required.
\end{proof}

\subsection{Induced Subgraphs of Schrijver Graphs}

We are ready to prove the lemmas that lie at the heart of the analysis of our algorithm for the $\SchrijverP$ problem.
The first lemma, given below, shows that in a large induced subgraph of the Schrijver graph $S(n,k)$ whose vertices do not have a popular element, a random pair of vertices forms an edge with a non-negligible probability.

\begin{lemma}\label{lemma:at_most_gamma}
For integers $k \geq 2$ and $n \geq 2k$, let $\calF$ be a family of stable $k$-subsets of $[n]$ whose size satisfies $|\calF| \geq k^2 \cdot \binom{n-k-2}{k-2}$ and let $\gamma \in (0,1]$.
Suppose that every element of $[n]$ belongs to at most $\gamma$ fraction of the sets of $\calF$.
Then, the probability that two random sets chosen uniformly and independently from $\calF$ are adjacent in $K(\calF)$ is at least
\[\frac{1}{2} \cdot \Biggr ( 1 - \gamma - \frac{k^2}{|\calF|} \cdot \binom{n-k-2}{k-2}\Biggr ) \cdot \Biggr ( 1 - \frac{k^2}{|\calF|} \cdot \binom{n-k-2}{k-2}\Biggr ).\]
\end{lemma}

\begin{proof}
Let $\calF \subseteq \binom{[n]}{k}_{\mathrm{stab}}$ be a family of sets as in the statement of the lemma.
We first claim that every subfamily $\calF' \subseteq \calF$ whose size satisfies
\begin{eqnarray}\label{eq:size_F'}
|\calF'| \geq \gamma \cdot |\calF| + k^2 \cdot \binom{n-k-2}{k-2}
\end{eqnarray}
spans an edge in $K(\calF)$.
To see this, consider such an $\calF'$, and notice that the assumption that every element of $[n]$ belongs to at most $\gamma$ fraction of the sets of $\calF$, combined with the fact that $|\calF'| > \gamma \cdot |\calF|$, implies that $\calF'$ is not a trivial family, that is, its sets do not share a common element.
In addition, using $|\calF'| > k^2 \cdot \binom{n-k-2}{k-2}$, it follows from Lemma~\ref{lemma:HMstable} that $\calF'$ is not a non-trivial intersecting family. We thus conclude that $\calF'$ is not an intersecting family, hence it spans an edge in $K(\calF)$.

We next show a lower bound on the size of a maximum matching in $K(\calF)$.
Consider the process that maintains a subfamily $\calF'$ of $\calF$, initiated as $\calF$, and that removes from $\calF'$ the two endpoints of some edge spanned by $\calF'$ as long as its size satisfies the condition given in~\eqref{eq:size_F'}. The pairs of vertices that are removed during the process form a matching $\calM$ in $K(\calF)$, whose size satisfies
\begin{eqnarray}\label{eq:matching}
|\calM| &\geq& \frac{1}{2} \cdot \Biggr (  |\calF| -  \Biggr (\gamma \cdot |\calF| + k^2 \cdot \binom{n-k-2}{k-2} \Biggr ) \Biggr ) \nonumber \\ &=& \frac{1}{2} \cdot \Biggr ( (1-\gamma) \cdot |\calF| - k^2 \cdot \binom{n-k-2}{k-2}\Biggr ).
\end{eqnarray}

We now consider the sum of the degrees of adjacent vertices in the graph $K(\calF)$.
Let $A,B \in \calF$ be any adjacent vertices in $K(\calF)$.
Since $A$ and $B$ are adjacent, they satisfy $A \cap B = \emptyset$, hence every vertex of $\calF$ that is not adjacent to $A$ nor to $B$ must include two distinct elements $a \in A$ and $b \in B$. For every two such elements, it follows from Claim~\ref{claim:stable_i_j} that the number of stable $k$-subsets of $[n]$ that include them both is at most $\binom{n-k-2}{k-2}$. Therefore, the number of vertices of $\calF$ that are not adjacent to $A$ nor to $B$ does not exceed $k^2 \cdot \binom{n-k-2}{k-2}$. This implies that the degrees of $A$ and $B$ in $K(\calF)$ satisfy
\[d(A)+d(B) \geq |\calF| - k^2 \cdot \binom{n-k-2}{k-2}.\]

Let $\calE$ denote the edge set of $K(\calF)$. We combine the above bound with the lower bound given in~\eqref{eq:matching} on the size of the matching $\calM$, to obtain that
\begin{eqnarray*}
2 \cdot |\calE| = \sum_{F \in \calF}{d(F)} &\geq&
\sum_{\{A,B\} \in \calM}{(d(A)+d(B))} \geq
|\calM| \cdot \Biggr (  |\calF| - k^2 \cdot \binom{n-k-2}{k-2}\Biggr )
\\ &\geq& \frac{1}{2} \cdot \Biggr ( (1-\gamma) \cdot |\calF| - k^2 \cdot \binom{n-k-2}{k-2}\Biggr ) \cdot \Biggr (  |\calF| - k^2 \cdot \binom{n-k-2}{k-2}\Biggr ).
\end{eqnarray*}
Finally, consider a pair of random sets chosen uniformly and independently from $\calF$.
The probability that they form an edge in $K(\calF)$ is twice the number of edges in $K(\calF)$ divided by $|\calF|^2$.
Hence, the above bound on $2 \cdot |\calE|$ completes the proof.
\end{proof}

As a corollary of Lemma~\ref{lemma:at_most_gamma}, we obtain the following.
\begin{corollary}\label{cor:at_most_gamma}
For integers $k \geq 2$ and $n \geq 8k^4$, let $\calF$ be a family of stable $k$-subsets of $[n]$ of size $|\calF| \geq \frac{1}{2n} \cdot \big |\binom{[n]}{k}_{\mathrm{stab}} \big |$ and let $\gamma \in (0,1]$.
Suppose that every element of $[n]$ belongs to at most $\gamma$ fraction of the sets of $\calF$.
Then, the probability that two random sets chosen uniformly and independently from $\calF$ are adjacent in $K(\calF)$ is at least $\frac{3}{8} \cdot (\frac{3}{4}-\gamma)$.
\end{corollary}

\begin{proof}
By Claim~\ref{claim:|V(S(n,k))|}, a family $\calF$ as in the statement of the corollary satisfies $|\calF| \geq \frac{1}{2k} \cdot \binom{n-k-1}{k-1}$.
This implies that
\begin{eqnarray*}
\frac{k^2}{|\calF|} \cdot \binom{n-k-2}{k-2} & \leq & 2k^3 \cdot \frac{\binom{n-k-2}{k-2}}{\binom{n-k-1}{k-1}} = \frac{2k^3 \cdot (k-1)}{n-k-1} \leq \frac{1}{4},
\end{eqnarray*}
where for the second inequality we use the assumption $n \geq 8k^4$.
Applying Lemma~\ref{lemma:at_most_gamma}, we obtain that the probability that two random sets chosen uniformly and independently from $\calF$ are adjacent in $K(\calF)$ is at least
$\frac{1}{2} \cdot (1-\gamma-\frac{1}{4}) \cdot (1-\frac{1}{4}) = \frac{3}{8} \cdot (\frac{3}{4}-\gamma)$.
\end{proof}

The following lemma shows that if a large collection of vertices of $S(n,k)$ has a quite popular element, then every $k$-subset of $[n]$ that does not include this element is disjoint from many of the vertices in the collection.

\begin{lemma}\label{lemma:at_least_gamma}
For integers $k \geq 2$ and $n \geq 2k$, let $X \subseteq [n]$ be a set, let $\calF \subseteq \binom{X}{k}_{\mathrm{stab}}$ be a family, and let $\gamma \in (0,1]$.
Let $j \in X$ be an element that belongs to at least $\gamma$ fraction of the sets of $\calF$, and suppose that $A \in \binom{[n]}{k}$ is a set satisfying $j \notin A$.
Then, the probability that a random set chosen uniformly from $\calF$ is disjoint from $A$ is at least
\[\gamma - \frac{k}{|\calF|} \cdot \binom{|X|-k-2}{k-2}.\]
\end{lemma}

\begin{proof}
Let $\calF \subseteq \binom{X}{k}_{\mathrm{stab}}$ be a family as in the lemma, and put $\calF' = \{F \in \calF \mid j \in F\}$.
By assumption, it holds that $|\calF'| \geq \gamma \cdot |\calF|$.
Suppose that $A \in \binom{[n]}{k}$ is a set satisfying $j \notin A$.
We claim that for every $i \in A$, the number of sets $B \in \binom{X}{k}_{\mathrm{stab}}$ satisfying $\{i,j\} \subseteq B$ does not exceed $\binom{|X|-k-2}{k-2}$.
Indeed, if $i \notin X$ then there are no such sets, so the bound trivially holds, and if $i \in X$, the bound follows from Claim~\ref{claim:stable_i_j}, using the one-to-one correspondence between $\binom{X}{k}_{\mathrm{stab}}$ and the vertex set of $S(|X|,k)$.
This implies that the number of sets $B \in \binom{X}{k}_{\mathrm{stab}}$ with $j \in B$ that intersect $A$ does not exceed $k \cdot \binom{|X|-k-2}{k-2}$.
It thus follows that the number of sets of $\calF$ that are disjoint from $A$ is at least
\[ |\calF'| - k \cdot \binom{|X|-k-2}{k-2} \geq \gamma \cdot |\calF| -k \cdot \binom{|X|-k-2}{k-2}.\]
Hence, a random set chosen uniformly from $\calF$ is disjoint from $A$ with the desired probability.
\end{proof}

As a corollary of Lemma~\ref{lemma:at_least_gamma}, we obtain the following.

\begin{corollary}\label{cor:at_least_gamma}
For integers $k \geq 2$ and $n$, let $X \subseteq [n]$ be a set of size $|X| \geq 16k^3$, let $\calF \subseteq \binom{X}{k}_{\mathrm{stab}}$ be a family of size $|\calF| \geq \frac{1}{2|X|} \cdot \big |\binom{X}{k}_{\mathrm{stab}} \big |$, and let $\gamma \in (0,1]$.
Let $j \in X$ be an element that belongs to at least $\gamma$ fraction of the sets of $\calF$, and suppose that $A \in \binom{[n]}{k}$ is a set satisfying $j \notin A$.
Then, the probability that a random set chosen uniformly from $\calF$ is disjoint from $A$ is at least $\gamma-\frac{1}{8}$.
\end{corollary}

\begin{proof}
Recall that there is a one-to-one correspondence between the collection $\binom{X}{k}_{\mathrm{stab}}$ and the vertex set of $S(|X|,k)$.
By Claim~\ref{claim:|V(S(n,k))|}, a family $\calF$ as in the statement of the corollary satisfies $|\calF| \geq \frac{1}{2k} \cdot \binom{|X|-k-1}{k-1}$.
This implies that
\begin{eqnarray*}
\frac{k}{|\calF|} \cdot \binom{|X|-k-2}{k-2} & \leq & 2k^2 \cdot \frac{\binom{|X|-k-2}{k-2}}{\binom{|X|-k-1}{k-1}} = \frac{2k^2 \cdot (k-1)}{|X|-k-1} \leq \frac{1}{8},
\end{eqnarray*}
where for the second inequality we use the assumption $|X| \geq 16k^3$.
Applying Lemma~\ref{lemma:at_least_gamma}, we obtain that the probability that a random set chosen uniformly from $\calF$ is disjoint from $A$ is at least $\gamma-\frac{1}{8}$, as desired.
\end{proof}

\subsection{Induced Subgraphs of Kneser Graphs}

The previous section shows several properties of induced subgraphs of Schrijver graphs.
We note that analogue results hold for Kneser graphs as well (see~\cite[Lemmas~3.1,~3.3 and Corollaries~3.2,~3.4]{Haviv22a}).
Moreover, the following proposition provides an analogue of Corollary~\ref{cor:at_most_gamma} for Kneser graphs, where the condition on $n$ and $k$ is somewhat weaker, namely, $n = \Omega(k^3)$ rather than $n = \Omega(k^4)$.

\begin{proposition}\label{prop:k^3}
There exists an integer $k_0$ such that for all integers $k \geq k_0$ and $n \geq 8k^3$ the following holds.
Let $\calF$ be a family of $k$-subsets of $[n]$ of size $|\calF| \geq \frac{1}{2n} \cdot \binom{n}{k}$ and let $\gamma \in [0,1]$. Suppose that every element of $[n]$ belongs to at most $\gamma$ fraction of the sets of $\calF$. Then, the probability that two random sets chosen uniformly and independently from $\calF$ are adjacent in $K(\calF)$ is at least $0.49 \cdot (1+\gamma) \cdot (\tfrac{3}{4}-\gamma)$.
\end{proposition}
\noindent
Proposition~\ref{prop:k^3} does not yield a fixed-parameter algorithm for the $\KneserP$ problem whose running time is asymptotically better than the running time of our algorithm for the $\SchrijverP$ problem. Yet, it can be used to improve the size of the instances produced by the corresponding kernelization algorithm for the $\KneserP$ problem (see Section~\ref{sec:kernel}).

The proof of Proposition~\ref{prop:k^3} relies on the following result of Frankl and Kupavskii~\cite{FrankK20}.
\begin{theorem}[\cite{FrankK20}]\label{thm:FK}
There exists an integer $k_0$ such that for all integers $k \geq k_0$ and $n \geq 64k^2$ the following holds.
For every family $\calF$ of $k$-subsets of $[n]$ with $|\calF| \geq k \cdot \binom{n-2}{k-2}$, if $\calF$ is not intersecting then the graph $K(\calF)$ has a vertex of degree at least $0.49 \cdot |\calF|$.
\end{theorem}

We start with the following lemma.

\begin{lemma}\label{lemma:k^3}
There exists an integer $k_0$ such that for all integers $k \geq k_0$ and $n \geq 64k^2$ the following holds.
Let $\calF$ be a family of $k$-subsets of $[n]$ and let $\gamma \in [0,1]$. Suppose that every element of $[n]$ belongs to at most $\gamma$ fraction of the sets of $\calF$. Then, the probability that two random sets chosen uniformly and independently from $\calF$ are adjacent in $K(\calF)$ is at least
\[0.49 \cdot \Biggr ( 1+\gamma + \frac{k}{|\calF|} \cdot \binom{n-2}{k-2} \Biggr ) \cdot \Biggr ( 1-\gamma - \frac{k}{|\calF|} \cdot \binom{n-2}{k-2} \Biggr ).\]
\end{lemma}

\begin{proof}
Let $\calF \subseteq \binom{[n]}{k}$ be a family as in the lemma.
We first claim that every subfamily $\calF' \subseteq \calF$ whose size satisfies
\begin{eqnarray}\label{eq:size_F'k^3}
|\calF'| \geq \gamma \cdot |\calF| + k \cdot \binom{n-2}{k-2}
\end{eqnarray}
is not intersecting.
To see this, consider such an $\calF'$, and notice that the assumption that every element of $[n]$ belongs to at most $\gamma$ fraction of the sets of $\calF$, using $|\calF'| > \gamma \cdot |\calF|$, implies that $\calF'$ is not a trivial family.
In addition, by the Hilton--Milner theorem (Theorem~\ref{thm:HM}; see Remark~\ref{remark:HM}), using $|\calF'| > k \cdot \binom{n-2}{k-2}$, it follows that $\calF'$ is not a non-trivial intersecting family. We thus conclude that $\calF'$ is not intersecting.

Consider now the process that maintains a subfamily $\calF'$ of $\calF$, initiated as $\calF$, and that as long as its size satisfies the condition given in~\eqref{eq:size_F'k^3}, removes from $\calF'$ a set that forms a vertex in the graph $K(\calF')$ whose degree is at least $0.49 \cdot |\calF'|$. The existence of such a set is guaranteed by Theorem~\ref{thm:FK}, because $n \geq 64 k^2$ and because $\calF'$ satisfies~\eqref{eq:size_F'k^3}, which implies that $|\calF'| > k \cdot \binom{n-2}{k-2}$ and that it is not intersecting.
The removal of the set from $\calF'$ decreases the number of edges in $K(\calF')$ by at least $0.49 \cdot |\calF'|$.
Letting $\calE$ denote the edge set of $K(\calF)$, it follows that
\begin{eqnarray*}
|\calE| & \geq & 0.49 \cdot |\calF| + 0.49 \cdot (|\calF|-1) + \cdots + 0.49 \cdot \Biggr ( \gamma \cdot |\calF| + k \cdot \binom{n-2}{k-2} \Biggr ) \\
& \geq & \frac{0.49}{2} \cdot \Biggr ( (1+\gamma) \cdot |\calF| + k \cdot \binom{n-2}{k-2} \Biggr ) \cdot \Biggr ( (1-\gamma) \cdot |\calF| - k \cdot \binom{n-2}{k-2} \Biggr ).
\end{eqnarray*}
Finally, consider a pair of random sets chosen uniformly and independently from $\calF$.
The probability that they form an edge in $K(\calF)$ is twice the number of edges in $K(\calF)$ divided by $|\calF|^2$.
Hence, the above bound on $|\calE|$ completes the proof.
\end{proof}

It is now easy to derive Proposition~\ref{prop:k^3}.

\begin{proof}[ of Proposition~\ref{prop:k^3}]
By combining the assumptions $n \geq 8k^3$ and $|\calF| \geq \frac{1}{2n} \cdot \binom{n}{k}$, it follows that
\[ \frac{k}{|\calF|} \cdot \binom{n-2}{k-2} \leq 2nk \cdot \frac{\binom{n-2}{k-2}}{\binom{n}{k}} = 2nk \cdot \frac{k(k-1)}{n(n-1)} \leq \frac{1}{4}.\]
The proof is then completed by applying Lemma~\ref{lemma:k^3}.
\end{proof}

\section{A Fixed-Parameter Algorithm for the \texorpdfstring{$\SchrijverP$}{Schrijver} Problem}\label{sec:algo}

In this section we present our randomized fixed-parameter algorithm for the $\SchrijverP$ problem.
We first describe the `element elimination' algorithm and then use it to obtain the final algorithm and to prove Theorem~\ref{thm:AlgoKneserNew}.

\subsection{The Element Elimination Algorithm}

The `element elimination' algorithm, given by the following theorem, will be used to repeatedly reduce the size of the ground set of a Schrijver graph while looking for a monochromatic edge.

\begin{theorem}\label{thm:one_step}
There exists a randomized algorithm that given integers $n$ and $k$, a set $X \subseteq [n]$ of size $|X| \geq 8 k^4$, a parameter $\eps >0$, a set of colors $C \subseteq [n-2k+1]$ of size $|C|=|X|-2k+1$, and an oracle access to a coloring $c: \binom{X}{k}_{\mathrm{stab}} \rightarrow [n-2k+1]$ of the vertices of $K(\binom{X}{k}_{\mathrm{stab}})$, runs in time $\poly(n,\ln (1/\eps))$ and returns, with probability at least $1-\eps$,
\begin{enumerate}[(a).]
  \item\label{output:a} a monochromatic edge of $K(\binom{X}{k}_{\mathrm{stab}})$, or
  \item\label{output:b} a vertex $A \in \binom{X}{k}_{\mathrm{stab}}$ satisfying $c(A) \notin C$, or
  \item\label{output:c} a color $i \in C$ and an element $j \in X$ such that for every $A \in \binom{[n]}{k}_{\mathrm{stab}}$ with $j \notin A$, a random vertex $B$ chosen uniformly from $\binom{X}{k}_{\mathrm{stab}}$ satisfies $c(B)=i$ and $A \cap B = \emptyset$ with probability at least $\frac{1}{9n}$.
\end{enumerate}
\end{theorem}

\begin{proof}
For integers $n$ and $k$, let $X \subseteq [n]$, $C \subseteq [n-2k+1]$, and $c: \binom{X}{k}_{\mathrm{stab}} \rightarrow [n-2k+1]$ be an input satisfying $|X| \geq 8 k^4$ and $|C|=|X|-2k+1$ as in the statement of the theorem.
It can be assumed that $k \geq 2$. Indeed, Theorem~\ref{thm:KneserS} guarantees that the graph $K(\binom{X}{k}_{\mathrm{stab}})$, which is isomorphic to $S(|X|,k)$, has either a monochromatic edge or a vertex whose color does not belong to $C$. Hence, for $k =1$, an output of type~\eqref{output:a} or~\eqref{output:b} can be found by querying the oracle for the colors of all the vertices in time polynomial in $n$.
For $k \geq 2$, consider the algorithm that given an input as above acts as follows (see Algorithm~\ref{alg:elimination}).

The algorithm first selects uniformly and independently $m$ random sets $A_1, \ldots, A_m \in \binom{X}{k}_{\mathrm{stab}}$ for $m=b \cdot n^2 \cdot \ln (n/\eps)$, where $b$ is some fixed constant to be determined later (see lines~\ref{line:pick_s}--\ref{line:samples}), and queries the oracle for their colors.
If the sampled sets include two vertices that form a monochromatic edge in $K(\binom{X}{k}_{\mathrm{stab}})$, then the algorithm returns such an edge (output of type~\eqref{output:a}; see line~\ref{line:output_a}).
If they include a vertex whose color does not belong to $C$, then the algorithm returns it (output of type~\eqref{output:b}; see line~\ref{line:output_b}).
Otherwise, the algorithm defines $i^* \in C$ as a color that appears on a largest number of sampled sets $A_t$ (see lines~\ref{line:alpha_s}--\ref{line:alpha_max}).
It further defines $j^* \in X$ as an element that belongs to a largest number of sampled sets $A_t$ with $c(A_t)=i^*$ (see lines~\ref{line:gamma_s}--\ref{line:gamma_max}).
Then, the algorithm returns the pair $(i^*,j^*)$ (output of type~\eqref{output:c}; see line~\ref{line:output_c}).

\begin{algorithm}[ht]
    \caption{Element Elimination Algorithm (Theorem~\ref{thm:one_step})}
    \textbf{Input:} Integers $n$ and $k \geq 2$, a set $X \subseteq [n]$ of size $|X| \geq 8 k^4$, a set of colors $C \subseteq [n-2k+1]$ of size $|C|=|X|-2k+1$, and an oracle access to a coloring $c: \binom{X}{k}_{\mathrm{stab}} \rightarrow [n-2k+1]$ of $K(\binom{X}{k}_{\mathrm{stab}})$.\\
    \textbf{Output:} \eqref{output:a} A monochromatic edge of $K(\binom{X}{k}_{\mathrm{stab}})$, or \eqref{output:b} a vertex $A \in \binom{X}{k}_{\mathrm{stab}}$ satisfying $c(A) \notin C$, or \eqref{output:c} a color $i \in C$ and an element $j \in X$ such that for every $A \in \binom{[n]}{k}_{\mathrm{stab}}$ with $j \notin A$, a random vertex $B \in \binom{X}{k}_{\mathrm{stab}}$ chosen uniformly satisfies $c(B)=i$ and $A \cap B = \emptyset$ with probability at least $\frac{1}{9n}$.
    \begin{algorithmic}[1]
        \State{$m \leftarrow b \cdot n^2 \cdot \ln (n/\eps)$ for a sufficiently large constant $b$}\label{line:pick_s}
        \State{pick uniformly and independently at random sets $A_1, \ldots, A_m \in \binom{X}{k}_{\mathrm{stab}}$}\label{line:samples}
        \ForAll{$t,t' \in [m]$}
            \If{$c(A_t) =c(A_{t'})$ and $A_t \cap A_{t'} = \emptyset$}
                \State{\textbf{return} $\{A_t,A_{t'}\}$}\label{line:output_a}\Comment{output of type~\eqref{output:a}}
            \EndIf
        \EndFor
        \ForAll{$t \in [m]$}
            \If{$c(A_t) \notin C$}
                \State{\textbf{return} $A_t$}\label{line:output_b}\Comment{output of type~\eqref{output:b}}
            \EndIf
        \EndFor
        \ForAll{$i \in C$}\label{line:alpha_s}
            \State{$\widetilde{\alpha}_i \leftarrow \frac{1}{m} \cdot  |\{ t \in [m] \mid c(A_t)=i\}|$}\label{line:alpha}
        \EndFor
        \State{$i^* \leftarrow $~an $i \in C$ with largest value of $\widetilde{\alpha}_{i}$}\label{line:alpha_max}
        \ForAll{$j \in X$}\label{line:gamma_s}
            \State{$\widetilde{\gamma}_{i^*,j} \leftarrow \frac{1}{m} \cdot  |\{ t \in [m] \mid c(A_t)=i^*~\mbox{and}~j \in A_t\}|$}\label{line:gamma}
        \EndFor
        \State{$j^* \leftarrow $~a $j \in X$ with largest value of $\widetilde{\gamma}_{i^*,j}$}\label{line:gamma_max}
        \State{\textbf{return} $(i^*,j^*)$}\label{line:output_c}\Comment{output of type~\eqref{output:c}}
    \end{algorithmic}
    \label{alg:elimination}
\end{algorithm}

The running time of the algorithm is clearly polynomial in $n$ and in $\ln(1/\eps)$.
An output of the algorithm is called valid if it satisfies one of the three conditions described in the statement of the theorem.
We turn to prove that for every input, the algorithm returns a valid output, of type~\eqref{output:a},~\eqref{output:b}, or~\eqref{output:c}, with probability at least $1-\eps$.
We start with the following lemma that shows that if the input coloring has a large color class with no popular element, then the algorithm returns a valid output of type~\eqref{output:a} with the desired probability.

\begin{lemma}\label{lemma:elimination_no_pop}
Suppose that the input coloring $c$ has a color class $\calF \subseteq \binom{X}{k}_{\mathrm{stab}}$ of size $|\calF| \geq \frac{1}{2|X|} \cdot \big |\binom{X}{k}_{\mathrm{stab}} \big |$ such that every element of $X$ belongs to at most half of the sets of $\calF$.
Then, Algorithm~\ref{alg:elimination} returns a monochromatic edge with probability at least $1-\eps$.
\end{lemma}

\begin{proof}
Let $\calF$ be as in the lemma.
Using the assumptions $k \geq 2$ and $|X| \geq 8 k^4$ and using the one-to-one correspondence between $\binom{X}{k}_{\mathrm{stab}}$ and the vertex set of $S(|X|,k)$, we can apply Corollary~\ref{cor:at_most_gamma} with $\gamma = \frac{1}{2}$ to obtain that two random sets chosen uniformly and independently from $\calF$ are adjacent in $K(\calF)$ with probability at least $\frac{3}{8} \cdot (\frac{3}{4}-\gamma) = \frac{3}{32}$.
Further, since the family $\calF$ satisfies $|\calF| \geq \frac{1}{2|X|} \cdot \big | \binom{X}{k}_{\mathrm{stab}} \big |$, a random vertex chosen uniformly from $\binom{X}{k}_{\mathrm{stab}}$ belongs to $\calF$ with probability at least $\frac{1}{2|X|}$. Hence, for two random vertices chosen uniformly and independently from $\binom{X}{k}_{\mathrm{stab}}$, the probability that they both belong to $\calF$ is at least $(\frac{1}{2|X|})^2$, and conditioned on this event, their probability to form an edge in $K(\calF)$ is at least $\frac{3}{32}$. This implies that the probability that two random vertices chosen uniformly and independently from $\binom{X}{k}_{\mathrm{stab}}$ form a monochromatic edge in $K(\binom{X}{k}_{\mathrm{stab}})$ is at least $(\frac{1}{2|X|})^2 \cdot \frac{3}{32} = \frac{3}{128|X|^2}$.

Now, by considering $\lfloor m/2 \rfloor$ pairs of the random sets chosen by Algorithm~\ref{alg:elimination} (line~\ref{line:samples}), it follows that the probability that no pair forms a monochromatic edge does not exceed
\[ \Big (1-\tfrac{3}{128|X|^2} \Big )^{\lfloor m/2\rfloor } \leq e^{-3 \cdot \lfloor m/2\rfloor / (128|X|^2)} \leq \eps,\]
where the last inequality follows by $|X| \leq n$ and by the choice of $m$, assuming that the constant $b$ is sufficiently large.
It thus follows that with probability at least $1-\eps$, the algorithm returns a monochromatic edge, as required.
\end{proof}

We next handle the case in which every large color class of the input coloring has a popular element.
To do so, we first show that the samples of the algorithm provide a good estimation for the fraction of vertices in each color class as well as for the fraction of vertices that share any given element in each color class.
For every color $i \in C$, let $\alpha_i$ denote the fraction of vertices of $K(\binom{X}{k}_{\mathrm{stab}})$ colored $i$, that is,
\[\alpha_i = \frac{|\{ A \in \binom{X}{k}_{\mathrm{stab}} \mid c(A)=i\}|}{\big | \binom{X}{k}_{\mathrm{stab}} \big |},\]
and let $\widetilde{\alpha}_i$ denote the fraction of the vertices sampled by the algorithm that are colored $i$ (see line~\ref{line:alpha}).
Similarly, for every $i \in C$ and $j \in X$, let $\gamma_{i,j}$ denote the fraction of vertices of $K(\binom{X}{k}_{\mathrm{stab}})$ colored $i$ that include $j$, that is,
\[\gamma_{i,j} = \frac{|\{ A \in \binom{X}{k}_{\mathrm{stab}} \mid c(A)=i~\mbox{and}~j \in A\}|}{ \big | \binom{X}{k}_{\mathrm{stab}} \big |},\]
and let $\widetilde{\gamma}_{i,j}$ denote the fraction of the vertices sampled by the algorithm that are colored $i$ and include $j$.
Let $E$ denote the event that
\begin{eqnarray}\label{eq:approx}
|\alpha_i - \widetilde{\alpha}_i| \leq \frac{1}{9|X|}~~\mbox{and}~~|\gamma_{i,j} - \widetilde{\gamma}_{i,j}| \leq \frac{1}{9|X|}~~\mbox{for all}~i \in C,~j \in X.
\end{eqnarray}

By a standard concentration argument, we obtain the following lemma.
\begin{lemma}\label{lemma:E}
The probability of the event $E$ is at least $1-\eps$.
\end{lemma}
\begin{proof}
By the Chernoff--Hoeffding bound (Theorem~\ref{thm:chernoff}) applied with $\mu=\frac{1}{9|X|}$, the probability that an inequality from~\eqref{eq:approx} does not hold is at most
\[2 \cdot e^{-2m/(81|X|^2)} \leq \frac{\eps}{n^2},\]
where the inequality follows by $|X| \leq n$ and by the choice of $m$, assuming that the constant $b$ is sufficiently large.
By the union bound over all the colors $i \in C$ and all the pairs $(i,j) \in C \times X$, that is, over $|C|+|C| \cdot |X| = |C| \cdot (1+|X|) \leq n^2$ events, we get that all the inequalities in~\eqref{eq:approx} hold with probability at least $1 - n^2 \cdot \frac{\eps}{n^2} = 1-\eps$, as required.
\end{proof}

We now show that if every large color class of the input coloring has a popular element and the event $E$ occurs, then the algorithm returns a valid output.

\begin{lemma}\label{lemma:elimination_pop}
Suppose that the coloring $c$ satisfies that for every color class $\calF \subseteq \binom{X}{k}_{\mathrm{stab}}$ whose size satisfies $|\calF| \geq \frac{1}{2|X|} \cdot \big | \binom{X}{k}_{\mathrm{stab}} \big |$, there exists an element of $X$ that belongs to more than half of the sets of $\calF$.
Then, if the event $E$ occurs, Algorithm~\ref{alg:elimination} returns a valid output.
\end{lemma}

\begin{proof}
Assume that the event $E$ occurs.
If Algorithm~\ref{alg:elimination} returns an output of type~\eqref{output:a} or~\eqref{output:b}, i.e., a monochromatic edge or a vertex whose color does not belong to $C$, then the output is verified before it is returned and is thus valid.
So suppose that the algorithm returns a pair $(i^*,j^*) \in C \times X$.
Recall that the color $i^*$ is defined by Algorithm~\ref{alg:elimination} as an $i \in C$ with largest value of $\widetilde{\alpha}_i$ (see line~\ref{line:alpha_max}).
Since the colors of all the sampled sets belong to $C$, it follows that $\sum_{i \in C}{\widetilde{\alpha}_i} = 1$, and thus
\begin{eqnarray}\label{eq:alpha_i_star}
\widetilde{\alpha}_{i^*} \geq \frac{1}{|C|} \geq \frac{1}{|X|},
\end{eqnarray}
where the last inequality follows by $|C| = |X|-2k+1 \leq |X|$.

Let $\calF$ be the family of vertices of $K(\binom{X}{k}_{\mathrm{stab}})$ colored $i^*$, i.e., $\calF = \{A \in \binom{X}{k}_{\mathrm{stab}} \mid c(A)=i^*\}$.
Since the event $E$ occurs (see~\eqref{eq:approx}), it follows from~\eqref{eq:alpha_i_star} that
\begin{eqnarray}\label{eq:alpha_i_no_tilde}
\alpha_{i^*} \geq \widetilde{\alpha}_{i^*} - \frac{1}{9|X|} \geq \frac{1}{|X|} - \frac{1}{9|X|} = \frac{8}{9|X|},
\end{eqnarray}
and thus
\[|\calF| = \alpha_{i^*} \cdot \Big | \binom{X}{k}_{\mathrm{stab}} \Big | \geq \frac{8}{9|X|} \cdot \Big | \binom{X}{k}_{\mathrm{stab}} \Big |.\]
Hence, by the assumption of the lemma, there exists an element $j \in X$ that belongs to more than half of the sets of $\calF$, that is, $\gamma_{i^*,j} > \frac{1}{2} \cdot \alpha_{i^*}$.
Since the event $E$ occurs, it follows that this $j$ satisfies $\widetilde{\gamma}_{i^*,j} > \frac{1}{2} \cdot \alpha_{i^*} - \frac{1}{9|X|}$.
Recalling that the element $j^*$ is defined by Algorithm~\ref{alg:elimination} as a $j \in X$ with largest value of $\widetilde{\gamma}_{i^*,j}$ (see line~\ref{line:gamma_max}), it must satisfy $\widetilde{\gamma}_{i^*,j^*} > \frac{1}{2} \cdot \alpha_{i^*} - \frac{1}{9|X|}$, and using again the fact that the event $E$ occurs, we derive that $\gamma_{i^*,j^*} \geq \widetilde{\gamma}_{i^*,j^*} -\frac{1}{9|X|} > \frac{1}{2} \cdot \alpha_{i^*} - \frac{2}{9|X|}$. This implies that $j^*$ belongs to at least $\gamma$ fraction of the sets of $\calF$ for $\gamma = \frac{\gamma_{i^*,j^*}}{\alpha_{i^*}} > \frac{1}{2} - \frac{2}{9 |X| \cdot \alpha_{i^*}} \geq \frac{1}{4}$, where the last inequality follows from~\eqref{eq:alpha_i_no_tilde}.

By $k \geq 2$ and $|X| \geq 8k^4 \geq 16 k^3$, we can apply Corollary~\ref{cor:at_least_gamma} with $\calF$, $j^*$, and $\gamma \geq \frac{1}{4}$ to obtain that for every set $A \in \binom{[n]}{k}_{\mathrm{stab}}$ with $j^* \notin A$, the probability that a random set chosen uniformly from $\calF$ is disjoint from $A$ is at least $\gamma-\frac{1}{8} \geq \frac{1}{8}$. Since the probability that a random set chosen uniformly from $\binom{X}{k}_{\mathrm{stab}}$ belongs to $\calF$ is at least $\frac{8}{9|X|}$, it follows that the probability that a random set $B$ chosen uniformly from $\binom{X}{k}_{\mathrm{stab}}$ satisfies $c(B) = i^*$ and $A \cap B = \emptyset$ is at least
\[\frac{8}{9|X|} \cdot \frac{1}{8} = \frac{1}{9|X|} \geq \frac{1}{9n}.\]
This implies that $(i^*,j^*)$ is a valid output of type~\eqref{output:c}.
\end{proof}

Equipped with Lemmas~\ref{lemma:elimination_no_pop},~\ref{lemma:E}, and~\ref{lemma:elimination_pop}, we are ready to derive the correctness of Algorithm~\ref{alg:elimination} and to complete the proof of Theorem~\ref{thm:one_step}.
If the input coloring $c$ has a color class $\calF \subseteq \binom{X}{k}_{\mathrm{stab}}$ of size $|\calF| \geq \frac{1}{2|X|} \cdot \big | \binom{X}{k}_{\mathrm{stab}} \big |$ such that every element of $X$ belongs to at most half of the sets of $\calF$, then, by Lemma~\ref{lemma:elimination_no_pop}, the algorithm returns with probability at least $1-\eps$ a monochromatic edge, i.e., a valid output of type~\eqref{output:a}.
Otherwise, the input coloring $c$ satisfies that for every color class $\calF \subseteq \binom{X}{k}_{\mathrm{stab}}$ of size $|\calF| \geq \frac{1}{2|X|} \cdot \big | \binom{X}{k}_{\mathrm{stab}} \big |$ there exists an element of $X$ that belongs to more than half of the sets of $\calF$. By Lemma~\ref{lemma:E}, the event $E$ occurs with probability at least $1-\eps$, implying by Lemma~\ref{lemma:elimination_pop} that with such probability, the algorithm returns a valid output.
It thus follows that for every input coloring the algorithm returns a valid output with probability at least $1-\eps$, and we are done.
\end{proof}

\subsection{The Fixed-Parameter Algorithm for the \texorpdfstring{$\SchrijverP$}{Schrijver} Problem}

We now present our fixed-parameter algorithm for the $\SchrijverP$ problem and complete the proof of Theorem~\ref{thm:AlgoKneserNew}.

\begin{proof}[ of Theorem~\ref{thm:AlgoKneserNew}]
Suppose that we are given, for integers $n$ and $k$ with $n \geq 2k$, an oracle access to a coloring $c: \binom{[n]}{k}_{\mathrm{stab}} \rightarrow [n-2k+1]$ of the vertices of the Schrijver graph $S(n,k)$.
It suffices to present an algorithm with success probability at least $1/2$, because the latter can be easily amplified by repetitions.
Our algorithm has two phases, as described below (see Algorithm~\ref{alg:Kneser}).

\begin{algorithm}[!htp]
    \caption{The Algorithm for the $\SchrijverP$ Problem (Theorem~\ref{thm:AlgoKneserNew})}
    \textbf{Input:} Integers $n,k$ with $n \geq 2k$ and an oracle access to a coloring $c: \binom{[n]}{k}_{\mathrm{stab}} \rightarrow [n-2k+1]$.\\
    \textbf{Output:} A monochromatic edge of $S(n,k)$.
    \begin{algorithmic}[1]
        \State{$s \leftarrow \max (n-8k^4,0)$,~~$X_0 \leftarrow [n]$,~~$C_0 \leftarrow [n-2k+1]$}\Comment{$|C_0|=|X_0|-2k+1$}
        \ForAll{$l = 0,1,\ldots,s-1$}\label{line:first_p}\Comment{first phase}
            \State{\textbf{call} Algorithm~\ref{alg:elimination} with $n$, $k$, $X_l$, $C_l$, $\eps = \frac{1}{4n}$ and with the restriction of $c$ to $\binom{X_l}{k}_{\mathrm{stab}}$}
            \If{Algorithm~\ref{alg:elimination} returns an edge $\{A,B\}$ with $c(A)=c(B)$}\label{line:a_s}\Comment{output of type~\eqref{output:a}}
                \State{\textbf{return} $\{A,B\}$}
            \EndIf\label{line:a_e}
            \If{Algorithm~\ref{alg:elimination} returns a vertex $A \in \binom{X_l}{k}_{\mathrm{stab}}$ with $c(A) = i_r \notin C_l$}\label{line:b_s}\label{line:B_t_1}\Comment{output of type~\eqref{output:b}}
                \ForAll{$t \in [18n]$}
                    \State{pick uniformly at random a set $B_t \in \binom{X_r}{k}_{\mathrm{stab}}$}
                    \If{$c(B_t) = i_r$ and $A \cap B_t = \emptyset$}
                        \State{\textbf{return} $\{A,B_t\}$}
                    \EndIf
                \EndFor
                \State{\textbf{return} `failure'}
            \EndIf\label{line:b_e}
            \If{Algorithm~\ref{alg:elimination} returns a pair $(i_l,j_l) \in C_l \times X_l$}\label{line:c_s}\Comment{output of type~\eqref{output:c}}
                \State{$X_{l+1} \leftarrow X_l \setminus \{j_l\},~~C_{l+1} \leftarrow C_l \setminus \{i_l\}$}\Comment{$|C_{l+1}|=|X_{l+1}|-2k+1$}
            \EndIf\label{line:c_e}
        \EndFor\label{line:first_p_end}
        \State{query the oracle for the colors of all the vertices of $K(\binom{X_s}{k}_{\mathrm{stab}})$}\Comment{second phase}
        \If{there exists a vertex $A \in \binom{X_s}{k}_{\mathrm{stab}}$ of color $c(A) = i_r \notin C_s$}\label{line:bb_s}
            \ForAll{$t \in [18n]$}
                \State{pick uniformly at random a set $B_t \in \binom{X_r}{k}_{\mathrm{stab}}$}
                \If{$c(B_t) = i_r$ and $A \cap B_t = \emptyset$}\label{line:B_t_2}
                    \State{\textbf{return} $\{A,B_t\}$}
                \EndIf
            \EndFor
            \State{\textbf{return} `failure'}\label{line:bb_e}
        \Else \State{find $A,B \in \binom{X_s}{k}_{\mathrm{stab}}$ satisfying $c(A)=c(B)$ and $A \cap B = \emptyset$}\label{line:kneser_s}\Comment{exist by Theorem~\ref{thm:KneserS}~\cite{LovaszKneser}}
            \State{\textbf{return} $\{A,B\}$}\label{line:kneser_e}
        \EndIf
    \end{algorithmic}
    \label{alg:Kneser}
\end{algorithm}

In the first phase, the algorithm repeatedly applies the `element elimination' algorithm given in Theorem~\ref{thm:one_step} (Algorithm~\ref{alg:elimination}).
Initially, we define
\[s=\max(n-8k^4,0),~~X_0 = [n],~~\mbox{and}~~C_0 = [n-2k+1].\]
In the $l$th iteration, $0 \leq l < s$, we call Algorithm~\ref{alg:elimination} with $n$, $k$, $X_l$, $C_l$, $\eps = \frac{1}{4n}$ and with the restriction of the given coloring $c$ to the vertices of $\binom{X_l}{k}_{\mathrm{stab}}$ to obtain with probability at least $1-\eps$,
\begin{enumerate}[(a).]
    \item\label{output:a1} a monochromatic edge $\{A,B\}$ of $K(\binom{X_l}{k}_{\mathrm{stab}})$, or
    \item\label{output:b1} a vertex $A \in \binom{X_l}{k}_{\mathrm{stab}}$ satisfying $c(A) \notin C_l$, or
    \item\label{output:c1} a color $i_l \in C_l$ and an element $j_l \in X_l$ such that for every $A \in \binom{[n]}{k}_{\mathrm{stab}}$ with $j_l \notin A$, a random vertex $B$ chosen uniformly from $\binom{X_l}{k}_{\mathrm{stab}}$ satisfies $c(B)=i_l$ and $A \cap B = \emptyset$ with probability at least $\frac{1}{9n}$.
\end{enumerate}
As will be explained shortly, if the output of Algorithm~\ref{alg:elimination} is of type~\eqref{output:a1} or~\eqref{output:b1} then we either return a monochromatic edge or declare `failure', and if the output is a pair $(i_l,j_l)$ of type~\eqref{output:c1} then we define $X_{l+1} = X_l \setminus \{j_l\}$ and $C_{l+1} = C_l \setminus \{i_l\}$ and, as long as $l < s$, proceed to the next call of Algorithm~\ref{alg:elimination}. Note that the sizes of the sets $X_l$ and $C_l$ are reduced by $1$ in every iteration, hence we maintain the equality $|C_l| = |X_l| - 2k+1$ for all $l$.
We now describe how the algorithm acts in the $l$th iteration for each type of output returned by Algorithm~\ref{alg:elimination}.

If the output is of type~\eqref{output:a1}, then the returned monochromatic edge of $K(\binom{X_l}{k}_{\mathrm{stab}})$ is also a monochromatic edge of $S(n,k)$, so we return it (see lines~\ref{line:a_s}--\ref{line:a_e}).

If the output is of type~\eqref{output:b1}, then we are given a vertex $A \in \binom{X_l}{k}_{\mathrm{stab}}$ satisfying $c(A) = i_r \notin C_l$ for some $r<l$. Since $i_r \notin C_l$, it follows that $j_r \notin X_l$, and thus $j_r \notin A$.
In this case, we pick uniformly and independently $18n$ random sets from $\binom{X_r}{k}_{\mathrm{stab}}$ and query the oracle for their colors.
If we find a vertex $B$ that forms together with $A$ a monochromatic edge in $S(n,k)$, we return the monochromatic edge $\{A,B\}$, and otherwise we declare `failure' (see lines~\ref{line:b_s}--\ref{line:b_e}).

If the output of Algorithm~\ref{alg:elimination} is a pair $(i_l,j_l)$ of type~\eqref{output:c1}, then we define, as mentioned above, the sets $X_{l+1} = X_l \setminus \{j_l\}$ and $C_{l+1} = C_l \setminus \{i_l\}$ (see lines~\ref{line:c_s}--\ref{line:c_e}).
Observe that for $0 \leq l < s$, it holds that $|X_{l}| = n-l > n-s = 8 k^4$, allowing us, by Theorem~\ref{thm:one_step}, to call Algorithm~\ref{alg:elimination} in the $l$th iteration.

In case that all the $s$ calls to Algorithm~\ref{alg:elimination} return an output of type~\eqref{output:c1}, we arrive to the second phase of the algorithm.
Here, we are given the sets $X_s$ and $C_s$ that satisfy $|X_s| = n-s \leq 8k^4$ and $|C_s| = |X_s|-2k+1$, and we query the oracle for the colors of each and every vertex of the graph $K(\binom{X_s}{k}_{\mathrm{stab}})$.
If we find a vertex $A \in \binom{X_s}{k}_{\mathrm{stab}}$ satisfying $c(A) = i_r \notin C_s$ for some $r<s$, then, as before, we pick uniformly and independently $18n$ random sets from $\binom{X_r}{k}_{\mathrm{stab}}$ and query the oracle for their colors. If we find a vertex $B$ that forms together with $A$ a monochromatic edge in $S(n,k)$, we return the monochromatic edge $\{A,B\}$, and otherwise we declare `failure' (see lines~\ref{line:bb_s}--\ref{line:bb_e}).
Otherwise, all the vertices of $K(\binom{X_s}{k}_{\mathrm{stab}})$ are colored by colors from $C_s$. By Theorem~\ref{thm:KneserS}, the chromatic number of the graph $K(\binom{X_s}{k}_{\mathrm{stab}})$, which is isomorphic to $S(|X_s|,k)$, is $|X_s|-2k+2 > |C_s|$. Hence, there must exist a monochromatic edge in $K(\binom{X_s}{k}_{\mathrm{stab}})$, and by checking all the pairs of its vertices we find such an edge and return it (see lines~\ref{line:kneser_s}--\ref{line:kneser_e}).

We turn to analyze the probability that Algorithm~\ref{alg:Kneser} returns a monochromatic edge.
Note that whenever the algorithm returns an edge, it ensures that it is monochromatic.
Hence, it suffices to show that the algorithm declares `failure' with probability at most $1/2$.
To see this, recall that the algorithm calls Algorithm~\ref{alg:elimination} at most $s < n$ times, and that by Theorem~\ref{thm:one_step} the probability that its output is not valid is at most $\eps = \frac{1}{4n}$. By the union bound, the probability that any of the calls to Algorithm~\ref{alg:elimination} returns an invalid output does not exceed $1/4$.
The only situation in which Algorithm~\ref{alg:Kneser} declares `failure' is when it finds, for some $r <s$, a vertex $A \in \binom{[n]}{k}_{\mathrm{stab}}$ with $c(A)=i_r$ and $j_r \notin A$, and none of the $18n$ sampled sets $B \in \binom{X_r}{k}_{\mathrm{stab}}$ satisfies $c(B)=i_r$ and $A \cap B = \emptyset$ (see lines~\ref{line:b_s}--\ref{line:b_e},~\ref{line:bb_s}--\ref{line:bb_e}).
However, assuming that all the calls to Algorithm~\ref{alg:elimination} return valid outputs, the $r$th run guarantees, by Theorem~\ref{thm:one_step}, that a random vertex $B$ uniformly chosen from $\binom{X_r}{k}_{\mathrm{stab}}$ satisfies $c(B)=i_r$ and $A \cap B = \emptyset$ for the given $A$ with probability at least $\frac{1}{9n}$.
Hence, the probability that the algorithm declares `failure' does not exceed $(1-\frac{1}{9n})^{18n} \leq e^{-2} < \frac{1}{4}$.
Using again the union bound, it follows that the probability that Algorithm~\ref{alg:Kneser} either gets an invalid output from Algorithm~\ref{alg:elimination} or fails to find a vertex that forms a monochromatic edge with a set $A$ as above is at most $1/2$.
Therefore, the probability that Algorithm~\ref{alg:Kneser} successfully finds a monochromatic edge is at least $1/2$, as desired.

We finally analyze the running time of Algorithm~\ref{alg:Kneser}.
In its first phase, the algorithm calls Algorithm~\ref{alg:elimination} at most $s < n$ times, where the running time needed for each call is, by Theorem~\ref{thm:one_step} and by our choice of $\eps$, polynomial in $n$. It is clear that the other operations made throughout this phase can also be implemented in time polynomial in $n$.
In its second phase, the algorithm enumerates all the vertices of $K(\binom{X_s}{k}_{\mathrm{stab}})$. This phase can be implemented in running time polynomial in $n$ and in the number of vertices of this graph. The latter is $\big | \binom{X_s}{k}_{\mathrm{stab}} \big | \leq |X_s|^k \leq (8k^4)^k =  k^{O(k)}$.
It thus follows that the total running time of Algorithm~\ref{alg:Kneser} is $n^{O(1)} \cdot k^{O(k)}$, completing the proof.
\end{proof}

\subsection{Turing Kernelization for the $\KneserP$ and $\SchrijverP$ Problems}\label{sec:kernel}

As mentioned in the introduction, our fixed-parameter algorithm for the $\SchrijverP$ problem can be viewed as a randomized polynomial Turing kernelization algorithm.
In what follows we extend on this aspect of the algorithm.

We start with the definition of a Turing kernelization algorithm as it is usually used in the area of parameterized complexity (see, e.g.,~\cite[Chapter~22]{KernelBook19}).
A decision parameterized problem $P$ is a language of pairs $(x,k)$ where $k$ is an integer referred to as the parameter of $P$.
For a decision parameterized problem $P$ with parameter $k$ and for a computable function $f$, a Turing kernelization algorithm of size $f$ for $P$ is an algorithm that decides whether an input $(x,k)$ belongs to $P$ in polynomial time given an oracle access that decides membership in $P$ for instances $(x',k')$ with $|x'| \leq f(k)$ and $k' \leq f(k)$. The kernelization algorithm is said to be polynomial if $f$ is a polynomial function.

The $\KneserP$ and $\SchrijverP$ problems, however, are total search problems whose input is given as an oracle access.
Hence, for an instance associated with integers $n$ and $k$, we require a Turing kernelization algorithm to find a solution using an oracle that finds solutions for instances associated with integers $n'$ and $k'$ that are bounded by some function of $k$. We further require the algorithm to be able to simulate the queries of the produced instances using queries to the oracle associated with the original input.
Note that two different types of oracles are involved here: the oracle that solves the $\KneserP$ (resp. $\SchrijverP$) problem on graphs $K(n',k')$ (reps. $S(n',k')$) with bounded $n',k'$, and the oracle that supplies an access to the instance of the $\KneserP$ (resp. $\SchrijverP$) problem.

We claim that the proof of Theorem~\ref{thm:AlgoKneserNew} shows that the $\SchrijverP$ problem admits a randomized polynomial Turing kernelization algorithm in the following manner.
Given an instance of the $\SchrijverP$ problem, i.e., a coloring $c: \binom{[n]}{k}_{\mathrm{stab}} \rightarrow [n-2k+1]$ of the vertices of a Schrijver graph $S(n,k)$ for integers $n$ and $k$ with $n \geq 2k$, the first phase of Algorithm~\ref{alg:Kneser} runs in time polynomial in $n$ and either finds a monochromatic edge or produces a ground set $X_s \subseteq [n]$ and a set of colors $C_s \subseteq [n-2k+1]$ satisfying $|C_s| = |X_s|-2k+1$ and $|X_s| = O(k^4)$ (see lines~\ref{line:first_p}--\ref{line:first_p_end}).
In the latter case, the restriction of the input coloring $c$ to the vertices of the graph $K(\binom{X_s}{k}_{\mathrm{stab}})$ is not guaranteed to use only colors from $C_s$, as required by the definition of the $\SchrijverP$ problem.
Yet, by applying the oracle to this restriction of $c$, simulating its queries using the access to the coloring $c$ of $S(n,k)$, we either get a monochromatic edge in $K(\binom{X_s}{k}_{\mathrm{stab}})$ or find a vertex whose color does not belong to $C_s$. In both cases, a solution to the original instance can be efficiently found with high probability.
Indeed, if a monochromatic edge is returned then it forms a monochromatic edge of $S(n,k)$.
Otherwise, as shown in the analysis of Algorithm~\ref{alg:Kneser}, a vertex $A \in \binom{X_s}{k}_{\mathrm{stab}}$ whose color does not belong to $C_s$ can be used to efficiently find with high probability a vertex $B \in \binom{[n]}{k}_{\mathrm{stab}}$ such that $\{A,B\}$ is a monochromatic edge (see lines~\ref{line:bb_s}--\ref{line:bb_e}).
Since the size of $X_s$ is bounded by $O(k^4)$, we refer to this kernelization algorithm as polynomial.

We finally mention that the approach of the proof of Theorem~\ref{thm:AlgoKneserNew} can also be used to directly provide an algorithm for the $\KneserP$ problem (see~\cite{Haviv22a}). Using Proposition~\ref{prop:k^3}, it can be shown that the $\KneserP$ problem admits a randomized polynomial Turing kernelization algorithm in the sense described above, with a slightly improved bound of $O(k^3)$ on the size of the ground set produced by the algorithm. To avoid repetitions, we omit the details.

\section{An Algorithm for the $\KneserP$ Problem Based on Schrijver Graphs}\label{sec:AlgoSchr}

In this section we present a simple deterministic algorithm for the $\KneserP$ problem and confirm Theorem~\ref{thm:KneserAlgSch}.

\begin{proof}[ of Theorem~\ref{thm:KneserAlgSch}]
Consider the algorithm that given integers $n$ and $k$ with $n \geq 2k$ and an oracle access to a coloring $c: \binom{[n]}{k} \rightarrow [n-2k+1]$ of the vertices of the Kneser graph $K(n,k)$, enumerates all vertices of the Schrijver graph $S(n,k)$, i.e., the sets of $\binom{[n]}{k}_{\mathrm{stab}}$, and queries the oracle for their colors. Then, the algorithm goes over all pairs of those vertices and returns a pair that forms a monochromatic edge. The existence of such an edge follows from Theorem~\ref{thm:KneserS}, which asserts that $\chi(S(n,k))= n-2k+2$. The running time of the algorithm is polynomial in the number of vertices of $S(n,k)$, which is equal, by Claim~\ref{claim:|V(S(n,k))|}, to $\frac{n}{k} \cdot \binom{n-k-1}{k-1} = \frac{n}{k} \cdot \binom{n-k-1}{n-2k} \leq n^{\min(k,n-2k+1)}$.
\end{proof}

\section{The $\Agree$ Problem}\label{sec:agree}

In this section we study the $\Agree$ problem.
After presenting its formal definition, we provide efficient algorithms for certain families of instances and explore its connections to the $\KneserP$ problem.

\subsection{The Definition}
For a collection $M$ of $m$ items, consider $\ell$ utility functions $u_i : P(M) \rightarrow \Q^{\geq 0}$ for $i \in [\ell]$ that map every subset of $M$ to a non-negative value. The functions $u_i$ are assumed to be monotone, that is, $u_i(A) \leq u_i(B)$ whenever $A \subseteq B$.
We refer to $u_i$ as the utility function associated with  the $i$th agent.
A set $S \subseteq M$ is said to be {\em agreeable} to agent $i$ if $u_i(S) \geq u_i(M \setminus S)$, that is, the $i$th agent values $S$ at least as much as its complement.
The following theorem was proved by Manurangsi and Suksompong~\cite{ManurangsiS19} (see~\cite{GoldbergHIMS20} for an alternative proof).
\begin{theorem}[\cite{ManurangsiS19}]\label{thm:MS}
For every collection $M$ of $m$ items and for every $\ell$ agents with monotone utility functions, there exists a set $S \subseteq M$ of size
\[ |S| \leq \min \Big ( \Big \lfloor \frac{m+\ell}{2} \Big \rfloor, m \Big )\]
that is agreeable to all agents.
\end{theorem}
\noindent
The $\Agree$ problem is defined as follows (see Remark~\ref{remark:syntactic}).
\begin{definition}\label{def:Agree}
In the $\Agree$ problem, given a collection $M$ of $m$ items and $\ell$ monotone utility functions $u_i : P(M) \rightarrow \Q^{\geq 0}$ for $i \in [\ell]$ associated with $\ell$ agents, the goal is to find a set $S \subseteq M$ of size $|S| \leq \min ( \lfloor \frac{m+\ell}{2} \rfloor, m )$ that is agreeable to all agents.
In the black-box input model, the utility functions are given as an oracle access that for $i \in [\ell]$ and $A \subseteq M$ returns $u_i(A)$.
In the white-box input model, the utility functions are given by a Boolean circuit that for $i \in [\ell]$ and $A \subseteq M$ computes $u_i(A)$.
\end{definition}
\noindent
By Theorem~\ref{thm:MS}, every instance of the $\Agree$ problem has a solution.
Note that for instances with $\ell \geq m$, the collection $M$ forms a solution.

\subsection{Algorithms for the $\Agree$ Problem}\label{sec:Algo_Agree}

Our algorithms for the $\Agree$ problem are obtained by combining our algorithms for the $\KneserP$ problem with the relation discovered in~\cite{ManurangsiS19} between the problems.
As mentioned before, using the proof idea of Theorem~\ref{thm:MS} in~\cite{ManurangsiS19}, it can be shown that the $\Agree$ problem is efficiently reducible to the $\KneserP$ problem.
For completeness, we describe below the reduction and some of its properties.
Note that we consider here the problems in the black-box input model, but the same reduction is applicable in the white-box input model as well (see Section~\ref{sec:models}).

\paragraph{The reduction.}
Consider an instance of the $\Agree$ problem, i.e., a collection $M=[m]$ of items and $\ell$ monotone utility functions $u_i : P(M) \rightarrow \Q^{\geq 0}$ for $i \in [\ell]$.
It can be assumed that $\ell < m$, as otherwise the collection $M$ forms a solution, and the reduction can produce an arbitrary instance.
Put $k = \lceil\frac{m-\ell}{2}\rceil$, and define a coloring $c: \binom{[m]}{k} \rightarrow [\ell]$ of the vertices of the Kneser graph $K(m,k)$ as follows. For a $k$-subset $A \subseteq [m]$, $c(A)$ is the smallest $i \in [\ell]$ satisfying $u_i(A) > u_i(M \setminus A)$ if such an $i$ exists, and $c(A) = \ell$ otherwise.
(Note that in the latter case, $c(A)$ can be defined arbitrarily, and we define it to be $\ell$ just for concreteness.)
The definition of $k$ implies that $m-2k+1 \geq m-2 \cdot (\frac{m-\ell+1}{2})+1 = \ell$, hence the coloring $c$ uses colors from $[m-2k+1]$, as required.

\paragraph{Correctness.}
For $\ell<m$, consider a solution to the constructed $\KneserP$ instance, that is, two disjoint $k$-subsets $A,B$ of $[m]$ satisfying $c(A)=c(B)$.
Consider the complement sets $M \setminus A$ and $M \setminus B$ whose size satisfies $|M \setminus A| = |M \setminus B| = m- k =\lfloor \frac{m+\ell}{2}\rfloor$.
We claim that at least one of them is agreeable to all agents and thus forms a solution to the given instance of the $\Agree$ problem.
To see this, let $i \in [\ell]$ denote the color of $A$ and $B$.
Suppose in contradiction that neither of the sets $M \setminus A$ and $M \setminus B$ is agreeable to all agents, hence each of them has an agent that prefers its complement. By the definition of the coloring $c$, it follows that $u_i(A) > u_i(M \setminus A)$ and $u_i(B) > u_i(M \setminus B)$. This implies that
\[ u_i(A) > u_i(M \setminus A) \geq  u_i(B) > u_i(M \setminus B) \geq u_i(A),\]
where the second and fourth inequalities follow from the monotonicity of the function $u_i$ and the fact that $A$ and $B$ are disjoint. This implies that $u_i(A) > u_i(A)$, a contradiction.

Note that given the vertices $A$ and $B$ of a monochromatic edge in the Kneser graph $K(m,k)$ it is possible to efficiently check which of the sets $M \setminus A$ and $M \setminus B$ forms a solution to the $\Agree$ instance by querying the oracle for the values of $u_i(A)$, $u_i(M \setminus A)$, $u_i(B)$, and $u_i(M \setminus B)$ for all $i \in [\ell]$.

\paragraph{Black-box model.}
The reduction is black-box in the following manner.
Given an oracle access to an instance of the $\Agree$ problem with parameters $\ell<m$, it is possible to simulate an oracle access to the above instance of the $\KneserP$ problem with parameters $m$ and $k = \lceil \frac{m-\ell}{2} \rceil$, such that any solution for the latter efficiently yields a solution for the former. Notice that the simulation of the oracle access to the coloring $c$ requires, for any given vertex $A \in \binom{[m]}{k}$, performing at most $2 \ell$ queries to the utility functions of the original instance to determine the values of $u_i(A)$ and $u_i(M \setminus A)$ for all $i \in [\ell]$.

\paragraph{}
The above reduction yields the following proposition.

\begin{proposition}\label{prop:reduction}
Suppose that there exists a (randomized) algorithm for the $\KneserP$ problem that given an oracle access to a coloring $c: \binom{n}{k} \rightarrow [n-2k+1]$ of the vertices of the Kneser graph $K(n,k)$ finds a monochromatic edge in running time $t(n,k)$ (with some given probability).
Then, there exists a (randomized) algorithm for the $\Agree$ problem that given an oracle access to an instance with $m$ items and $\ell$ agents ($\ell<m$), returns a solution in running time $t(m, \lceil \frac{m-\ell}{2} \rceil) \cdot m^{O(1)}$ (with the same probability).
\end{proposition}

Our first algorithm for the $\Agree$ problem, given in Theorem~\ref{thm:AlgoIntro1}, is obtained by combining Theorem~\ref{thm:AlgoKneserNew} and Proposition~\ref{prop:reduction}.

\begin{proof}[ of Theorem~\ref{thm:AlgoIntro1}]
By Theorem~\ref{thm:AlgoKneserNew}, there exists a randomized algorithm for the $\KneserP$ problem that given a coloring of $K(n,k)$ with $n-2k+1$ colors runs in time $t(n,k) \leq n^{O(1)} \cdot k^{O(k)}$ and finds a monochromatic edge with high probability.
Combining this algorithm with Proposition~\ref{prop:reduction}, it follows that there exists a randomized algorithm for the $\Agree$ problem that given an instance with $m$ items and $\ell$ agents ($\ell<m$), returns a solution in time $t(m,k) \cdot m^{O(1)} \leq m^{O(1)} \cdot k^{O(k)}$ for $k = \lceil \frac{m-\ell}{2}\rceil$ with high probability.
\end{proof}

As a corollary, we derive that the problem can be solved in polynomial time when the number of agents $\ell$ is not much smaller than the number of items $m$.
\begin{corollary}\label{cor:m<=l+log/loglog}
For every constant $\tilde{c} >0$, there exists a randomized algorithm that given an oracle access to an instance of the $\Agree$ problem with $m$ items and $\ell$ agents such that \[\ell \geq m - \tilde{c} \cdot \frac{\log m}{\log \log m}\]
runs in time polynomial in $m,\ell$ and returns a solution with high probability.
\end{corollary}

\begin{proof}
Fix a constant $\tilde{c} >0$, and consider an instance of the $\Agree$ problem with $m$ items and $\ell$ agents such that $\ell \geq m - \tilde{c} \cdot \frac{\log m}{\log \log m}$. Consider the algorithm that for $\ell \geq m$, returns the collection of all items, and for $\ell < m$, calls the algorithm given in Theorem~\ref{thm:AlgoIntro1}.
Observe that in the latter case, the integer $k = \lceil \frac{m-\ell}{2}\rceil$ satisfies $k^{O(k)} = m^{O(1)}$, hence the proof is completed.
\end{proof}

Our second algorithm for the $\Agree$ problem, given in Theorem~\ref{thm:AlgoIntro2}, is obtained by combining Theorem~\ref{thm:KneserAlgSch} and Proposition~\ref{prop:reduction}.

\begin{proof}[ of Theorem~\ref{thm:AlgoIntro2}]
By Theorem~\ref{thm:KneserAlgSch}, there exists an algorithm for the $\KneserP$ problem that given a coloring of $K(n,k)$ with $n-2k+1$ colors runs in time $t(n,k)$ and finds a monochromatic edge, where $t(n,k)$ is polynomial in $\frac{n}{k} \cdot \binom{n-k-1}{k-1} \leq n^{\min(k,n-2k+1)}$.
Combining this algorithm with Proposition~\ref{prop:reduction}, it follows that there exists an algorithm for the $\Agree$ problem that given an instance with $m$ items and $\ell$ agents ($\ell<m$), returns a solution in time polynomial in $t(m,k)$ for $k = \lceil \frac{m-\ell}{2}\rceil$, as required.
\end{proof}

As a corollary, we derive that the problem can be solved in polynomial time when the number of agents $\ell$ is a fixed constant.
\begin{corollary}\label{cor:constant_l}
For every constant $\ell$, there exists an algorithm that given an oracle access to an instance of the $\Agree$ problem with $m$ items and $\ell$ agents finds a solution in time polynomial in $m$.
\end{corollary}

\begin{proof}
Fix a constant $\ell \geq 1$, and consider an instance of the $\Agree$ problem with $m$ items and $\ell$ agents. Consider the algorithm that for $\ell \geq m$, returns the collection of all items, and for $\ell < m$, calls the algorithm given in Theorem~\ref{thm:AlgoIntro2}.
Observe that in the latter case, the integer $k = \lceil \frac{m-\ell}{2}\rceil$ satisfies $m-2k+1 \leq \ell+1$, hence the proof is completed.
\end{proof}

\subsection{Subset Queries}

In what follows, we consider a variant of the $\KneserP$ problem in the white-box input model, where the Boolean circuit that represents the coloring of the graph allows an extended type of queries, called {\em subset queries}.
Here, for a coloring $c$ of $K(n,k)$ with $n-2k+1$ colors, the input of the circuit is a pair $(i,B)$ of a color $i \in [n-2k+1]$ and a set $B \subseteq [n]$, and the circuit returns on $(i,B)$ whether $B$ contains a vertex colored $i$.
Note that this circuit enables to efficiently determine the color $c(A)$ of a given vertex $A$ by computing the output of the circuit on all the pairs $(i,A)$ with $i \in [n-2k+1]$. Moreover, by computing the output of the circuit on the pair $(c(A),[n]\setminus A)$, one can determine whether $A$ lies on a monochromatic edge. It can further be seen that if $A$ does lie on a monochromatic edge then the other endpoint of such an edge can be found efficiently.

\begin{remark}\label{remark:syntactic}
The existence of solutions to instances of problems in the complexity class $\TFNP$ should follow from a {\em syntactic} guarantee of totality. However, the Boolean circuit that represents a coloring with subset queries in the $\KneserP$ problem is not syntactically guaranteed to be consistent. This issue can be addressed by allowing in the definition of the problem a solution that demonstrates a violation of the input coloring.
In fact, such a modification is also needed in the definition of the $\Agree$ problem (see Definition~\ref{def:Agree}). There, one has to allow a solution that demonstrates a violation of the monotonicity of an input utility function.
For simplicity of presentation, we skip this issue in the proofs.
\end{remark}

We turn to prove Theorem~\ref{thm:KneserVSAgree}, which implies that the $\Agree$ problem is at least as hard as the $\KneserP$ problem with subset queries.
We consider here the problems in the white-box input model (see Section~\ref{sec:models} and Definitions~\ref{def:KneserSProblems} and~\ref{def:Agree}).

\paragraph{The reduction.}
Consider an instance of the $\KneserP$ problem, a coloring $c: \binom{[n]}{k} \rightarrow [n-2k+1]$ of the Kneser graph $K(n,k)$ for integers $n$ and $k$ satisfying $n \geq 2k$.
We define an instance of the $\Agree$ problem that consists of a collection $M = [n]$ of $n$ items and $\ell = n-2k+1$ agents. The utility function $u_i : P(M) \rightarrow \{0,1\}$ of the $i$th agent is defined for $i \in [\ell]$ as follows. For every $S \subseteq M$, $u_i(S) = 1$ if there exists a $k$-subset $A$ of $S$ such that $c(A)=i$, and $u_i(S)=0$ otherwise.
Observe that each function $u_i$ is monotone, because if a set $S$ contains a vertex of $K(n,k)$ colored $i$ so does any set that contains $S$.

\paragraph{Correctness.}
To prove the correctness of the reduction, consider a solution to the constructed instance of the $\Agree$ problem, i.e., a set $S \subseteq M$ of size
\[|S| \leq \min \Big ( \Big \lfloor \frac{n+\ell}{2} \Big \rfloor, n \Big ) = \min (n-k,n) = n-k \]
that is agreeable to all agents.
Note that its complement $M \setminus S$ has size at least $k$. Let $A$ be a $k$-subset of $M \setminus S$, and let $i$ denote the color of $A$ according to the coloring $c$.
By the definition of $u_i$, it holds that $u_i(M \setminus S)=1$, and since $S$ is agreeable to agent $i$, it follows that $u_i(S)=1$ as well. Hence, there exists a $k$-subset $B$ of $S$ such that $c(B)=i$, implying that the vertices $A$ and $B$ form a monochromatic edge in $K(n,k)$.

We claim that given the solution $S$ of the constructed $\Agree$ instance, it is possible to find in polynomial time a solution to the original instance of the $\KneserP$ problem. Indeed, given the solution $S$ one can define $A$ as an arbitrary $k$-subset of $M \setminus S$ and determine its color $i$ using the given Boolean circuit. As shown above, it follows that $u_i(S)=1$.
Then, it is possible to efficiently find a $k$-subset $B$ of $S$ with $c(B)=i$. To do so, we maintain a set initiated as $S$, and remove elements from the set as long as it still contains a vertex of color $i$ (which can be checked in polynomial time using the Boolean circuit), until we get the desired set $B$ of size $k$.
This gives us the monochromatic edge $\{A,B\}$ of the $\KneserP$ instance.

\paragraph{White-box model.}
It follows from the description of the reduction that the Boolean circuit that represents the input coloring with subset queries precisely represents the utility functions of the constructed instance of the $\Agree$ problem. In particular, the reduction can be implemented in polynomial time.

\section*{Acknowledgements}

We would like to thank Gabriel Istrate and Andrey Kupavskii for helpful discussions and the anonymous referees of this paper and of its preliminary versions~\cite{Haviv22a,Haviv22b} for their valuable suggestions and comments.

\bibliographystyle{alpha}
\bibliography{mono_kneser}

\end{document}